\documentclass[11pt]{article}

\usepackage{preamble}
\usepackage{tikz}
\usetikzlibrary{calc}

\newif\ifEXTABSTRACT
\EXTABSTRACTfalse		% <- Change here.
\newif\ifFULL
\FULLtrue			% <- Change here.

\ifEXTABSTRACT
\newcommand{\eav}[1]{#1}
\else
\newcommand{\eav}[1]{}
\fi

\ifFULL
\newcommand{\fv}[1]{#1}
\else
\newcommand{\fv}[1]{}
\fi

\usepackage{comment}

\begin{document} %%%
%%%%%%%%%%%%%%%%%%%%
%%%%%%%%%%%%%%%%%%%%
%============%
% Title Page %
%============%

\fv{
\title{\textbf{Secure Software Leasing \\ Without Assumptions}}
}

\eav{\title{\vspace{-1cm}\textbf{Secure Software Leasing Without Assumptions}}}

\author{
\phantom{XXXXXXXX}Anne Broadbent \thanks{University of Ottawa, Ottawa, Canada. \{\texttt{abroadbe,slord050,spodder\}@uottawa.ca} }
\and Stacey Jeffery \thanks{QuSoft and CWI, Amsterdam, Netherlands. \texttt{jeffery@cwi.nl}}
\and S\'ebastien Lord \footnotemark[1]\phantom{XXXXXXXX}
\and Supartha Podder \footnotemark[1]
\and Aarthi Sundaram \thanks{Microsoft Quantum, Redmond, USA. \texttt{aarthi.sundaram@microsoft.com}}
}

%\date{\today}

%\date{\today}
\date{}

\maketitle
\begin{abstract}
\eav{\vspace{-.75cm}}
Quantum cryptography is known for enabling functionalities that are unattainable using classical information alone.
  Recently, \emph{Secure Software Leasing (SSL)} has emerged as
 one of these areas of interest. Given a target circuit~$C$ from a circuit class, SSL produces an encoding of~$C$ that enables a recipient to evaluate $C$, and also enables the originator of the software to \emph{verify}  that the software has been \emph{returned} --- meaning that the recipient has relinquished the possibility of any further use of the software. Clearly,
 such a functionality is unachievable using classical information alone, since it is impossible to prevent a user from keeping a copy of the software. Recent results  have
shown the achievability of SSL using quantum information for a class of functions called \emph{compute-and-compare} (these are a generalization of the well-known \emph{point functions}).
These prior works, however all make use of setup or computational assumptions.
Here, we show that SSL is achievable for compute-and-compare circuits \emph{without any assumptions}.

Our technique involves the study of \emph{quantum copy-protection}, which is a notion related to~SSL, but where the encoding procedure inherently \emph{prevents} a would-be quantum software pirate from \emph{splitting} a single copy of an encoding for~$C$ into two parts, each of which enables a user to evaluate~$C$.
We show that point functions can be copy-protected \emph{without any assumptions}, for a novel security definition involving one  honest and one malicious evaluator; this is achieved  by showing that from any quantum message authentication code, we can derive such an  \emph{honest-malicious} copy-protection scheme.
We then show  that a generic honest-malicious copy-protection scheme 
implies SSL; by prior work, this yields SSL for compute-and-compare functions. 
\end{abstract}

%%%%%%%%%%%
% Content %
%%%%%%%%%%%

\fv{
%=======================%
\section{Introduction} %
%=======================%
\label{sec:intro}
}

One of the defining features of quantum information is the \emph{no-cloning} principle, according to which it is not possible, in general, to take an arbitrary quantum state and produce two copies of it~\cite{Par70,WZ82,Die82}. This principle is credited for many
of the feats of quantum information in cryptography, including quantum key distribution (QKD)~\cite{BB84} and quantum money~\cite{Wie83}. \fv{(For a survey on quantum cryptography, see \cite{BS16}).}

The quantum no-cloning principle tells us that, in a certain sense, quantum information behaves more like a \emph{physical} object than a digital one: there are situations where  quantum information can be distributed and used, but it cannot be duplicated. One such example is quantum money~\cite{Wie83}, in which
a quantum system is used to encode a very basic type of information --- the ability to verify authenticity.
However, we can envisage quantum encodings that achieve richer levels of applicability. We thus define a hierarchy of ``uncloneable'' objects, where the basic notion provides only authenticity, and the topmost notion provides \emph{functionality}.
The uncloneability hierarchy includes:
\begin{itemize}
\item \textbf{Authenticity.} In the first (most basic) level, the uncloneability property can be used to \emph{verify} authenticity.
\item \textbf{Information.} Next, \emph{information} is made uncloneable, meaning that there is some underlying data that can be decoded, but there are limitations on the possibility of copying this data while it is encoded.
\item \textbf{Functionality.} At the top level of the hierarchy, a \emph{functionality} is made uncloneable, meaning that there are limitations on how many users can simultaneously evaluate the functionality.
\end{itemize}
\fv{
For both, the case of \emph{information} and \emph{functionality}, a type of \emph{verification} is possible (but optional): this verification is a way to confirm that a message or functionality is returned; after such verification is confirmed, further reading/use of the encoded information is impossible.
}

We emphasize that none of the  concepts in the hierarchy are possible in a conventional digital world, since classical  information can be copied. Thus the hierarchy is best understood intuitively at the level of a physical analogy where, for example, authenticity is verified by physical objects and functionalities are distributed in \emph{hardware} devices.

\paragraph{Achieving the hierarchy.}
We summarize below the known results on achievability of the hierarchy.
\begin{enumerate}
\item The \emph{authenticity} level of the hierarchy is the most well-understood, and it includes quantum money \cite{Wie83}, quantum coins \cite{MS10}, and publicly-verifiable quantum money \cite{AC12}.
\item Next, the \emph{information} level  includes  \emph{tamper-evident} encryption \cite{Got03} and \emph{uncloneable encryption} \cite{BL20}.
We comment here on a technique of Gottesman~\cite{Got03} that is relevant to our work. In \cite{Got03}, it is shown that tamper-evident encryption can be achieved  using the primitive of  \emph{Quantum Message Authentication (QMA)}~\cite{BCG+02} --- in other words, the \emph{verification} of quantum authentication not only gives a guarantee that the underlying plaintext is intact, but \emph{also} that no adversary can gain information on the plaintext, \emph{even if the key is revealed}. Uncloneable encryption is a notion that is complementary to tamper-evident encryption, and it focuses on \emph{preventing} duplication of an underlying plaintext. In \cite{BL20}, it is shown to be achievable in the Quantum Random Oracle Model (QROM).

\item  Finally, the \emph{functionality} level of the hierarchy was first discussed in terms of \emph{quantum copy protection} by Aaronson~\cite{Aar09}: here, a quantum encoding allows the evaluation of a function on a chosen input, but in a way that the number of \emph{simultaneous} evaluations is limited. In \cite{Aar09}, copy protection for a class of functions is shown to exist assuming a quantum oracle; this was improved (for a more restricted family of circuits) to a \emph{classical} oracle in \cite{ALLZZ20arxiv}. Further work in~\cite{CMP20} improved the assumption to the QROM.\footnote{This is an improvement, since a QROM does not depend on the circuit to be computed.}
    A related concept, also at the functionality level of the hierarchy, was recently put forward:
    \emph{Secure Software Leasing (SSL)}, where a quantum encoding allows evaluation of a circuit, while also enabling the originator to verify that the software is \emph{returned} (meaning that it can no longer be used to evaluate the function).  SSL was  first studied by
Ananth and La~Placa~\cite{AL20arxiv}, where it was shown that SSL could be achieved for \emph{searchable compute-and-compare circuits}\footnote{A circuit class $\mathcal{C}$ is a  \emph{compute-and-compare} circuit class if for every circuit in $\mathcal{C}$, there is an associated circuit~$C$ and string~$\alpha$ such that on input~$x$, the circuit outputs 1 if and only if $C(x) = \alpha$. \emph{Searchability} refers to the fact that there is an efficient algorithm that, on input $C \in \mathcal{C}$,  outputs an~$x$ such that $C(x)=\alpha$. From this point on, \emph{searchability} is an implicit assumption throughout this work.}; in order to achieve their result (which is with respect to an \emph{honest} evaluation), they make use of strong cryptographic assumptions: quantum-secure subspace obfuscators, a common reference string, and the difficulty of the Learning With Errors (LWE) problem.  Further work~\cite{CMP20}
improved the result on achievability for the same class of circuits, this time against \emph{malicious evaluations}, and in the QROM. Very recently,
\cite{KNY20arxiv} showed the achievability of SSL, based on LWE, against honest evaluators, and for classes of functions beyond \emph{evasive} functions.\footnote{Informally, \emph{evasive}  functions are the class of functions such that it is hard to find an accepting input, given only black-box access to a functions. Note that compute-and-compare functions are evasive.}

\end{enumerate}

\fv{
\subsection{Summary of Contributions}
\label{sec:intro-contributions}}
\eav{

{\vspace{.5cm}
\noindent\textbf{\hspace{-.2cm} Our contribution:} We solve two important open problems related to SSL and quantum copy protection.\looseness=-1}}

\fv{
Due to their foundational role in the study of uncloneability as well as for potential applications, SSL and copy protection are emerging as  important elements of quantum cryptography.  In this work, we solve two important open problems related to SSL and quantum copy protection.
}

\paragraph{Secure Software Leasing.}

We show how to construct an SSL scheme for compute-and-compare circuits, against a malicious evaluator.
  Ours is the first scheme that makes no assumptions \fv{--- there are no setup assumptions, such as the QROM or a common reference string and no computational assumptions, such as one-way functions or the LWE assumption}. We thus show for the first time that SSL is achievable, unconditionally. A compromise we make in order to achieve this is the use of a natural but weaker notion of correctness \emph{with respect to a distribution}.
We note that general SSL was shown to be impossible~\cite{AL20arxiv}, and that~\cite{Aar09} mentions how \emph{learnable} functions cannot be copy protected. It is thus natural that we focus our efforts on achieving SSL for compute-and-compare circuits, which is a family of functions that is not learnable.

In more detail, we follow the security notion of~\cite{CMP20}, which postulates a game between a challenger, and  a pirate Pete. Upon sampling a circuit  from a given distribution, the challenger encodes the circuit and sends it to Pete. Pete then produces a register that he returns to the challenger who performs a \emph{verification}; upon successful verification, we continue the game (otherwise, we abort), by presenting to Pete a challenge input~$x \in \{0,1\}^n$ (chosen according to a given distribution).
The scheme is \emph{$\epsilon$-secure} if we can bound the probability that Pete correctly evaluates the circuit on the challenge input~$x$, to be within $\epsilon$
of his trivial guessing probability. Here, trivially guessing means that Pete answers the challenge by seeing only $x$ i.e., disregarding all other information obtained by interacting with the challenger. Thus, security is defined relative to the distribution on the circuits and on the challenges.
For SSL, $\eta$-correctness  is defined  with respect to an input distribution, and means that, up to some error term~$\eta$,  the honest evaluation on an encoded circuit produces the correct outcome, \emph{in expectation}.\footnote{\label{footnote:worst-case}This notion is weaker than the more common notion of correctness that holds for \emph{all} inputs. However, in \Cref{sec:generic},  we give evidence that achieving this stronger notion of correctness may be possible,  by showing that for the standard notion of copy-protection (against two malicious evaluators), correctness in expectation implies worst-case correctness, which would then imply worst-case correctness for SSL.}

We show how to achieve SSL with respect to the  uniform distribution on point functions, and the challenge distribution which samples uniformly from the distribution where the correct response is~$0$ or~$1$ (with equal probability) ---  denoted $\Dhalf_p$. Our technique is a reduction from SSL to \emph{honest-malicious} copy protection, as well as a new construction for quantum honest-malicious copy protection (with respect to essentially the same distributions as stated above). Prior work noted, informally, that copy protection implies SSL~\cite{AL20arxiv}. Here, we formally show that our new and weaker (and thus easier-to-achieve) notion of copy protection (see below)  implies~SSL. Our work focuses on achieving SSL for point functions; by applying our result with~\cite[Theorem 6]{CMP20} this implies SSL for compute-and-compare circuits.

\paragraph{Honest-Malicious Copy Protection.}
We define a new security model for copy protection: \emph{honest-malicious} copy protection. Here, we consider a game between a challenger, a pirate (Pete), and two evaluators. Importantly, the first evaluator, Bob, is \emph{honest} (meaning that he will execute the legitimate evaluation procedure) and the second evaluator, Charlie, is \emph{malicious}. In copy protection, we want to bound the probability that, after each receiving a quantum register from Pete, who takes as input a single copy protected program, the two evaluators (who cannot communicate), are \emph{both} able to correctly evaluate the encoded circuit. Following \cite{CMP20}, this is formalized by a game, parameterized by a distribution on the input circuits, and a corresponding  challenge distribution on pairs of $n$-bit strings. A challenger samples a circuit, encodes it using the copy protection scheme and sends the encoding to Pete who creates the two registers. Then a challenge pair $(x_1,x_2)$ is sampled from the challenge distribution; Bob receives~$x_1$ while Charlie receives $x_2$. They \emph{win} if they each produce the correct output of the original circuit evaluated on $x_1$ and~$x_2$, respectively. An honest-malicious copy protection scheme is \emph{$\epsilon$-secure} for the given distributions if the probability that the evaluators win the game is within $\epsilon$
 of the success probability of the trivial strategy that is achievable when Bob gets the full encoding and Charlie guesses to the best of his ability without interacting with Pete. As in the case of SSL, $\eta$-correctness for copy protection is defined with respect to an input distribution, and means that, up to some constant $\eta$,  the honest evaluation on an encoded circuit produces the correct outcome, \emph{in expectation}\footnote{See \Cref{footnote:worst-case}.}.

We establish the relevance of honest-malicious copy protection by showing that, for general functions, honest-malicious copy protection implies SSL.

\fv{
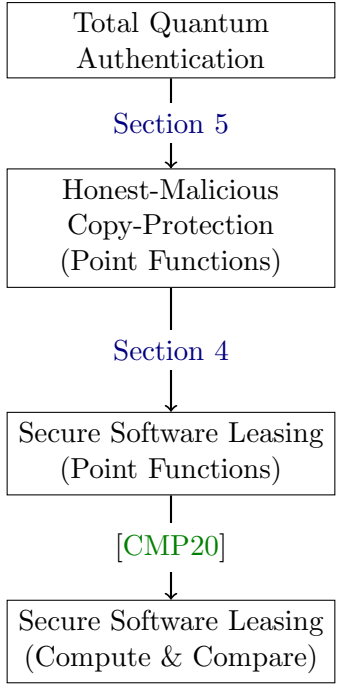
\begin{figure}
\begin{center}
\begin{tikzpicture}

\node[draw,align=center,text width=4.1cm] (HMCPPF)
	at ( 0,  0    )
	{Honest-Malicious Copy-Protection \\ (Point Functions)};

\node[draw,align=center,text width=4.1cm] (QAS)
	at ($ (HMCPPF) + (0,2.5) $)
	{Total Quantum Authentication};

\node[draw,align=center,text width=4.1cm] (SSLPF)
	at ($ (HMCPPF) + (0,-3) $)
	{Secure Software Leasing \\ (Point Functions)};

\node[draw,align=center,text width=4.1cm] (SSLCC)
	at ($ (SSLPF) + (0,-2.5) $)
	{Secure Software Leasing \\ (Compute \& Compare)};

\draw[->,thick]
	(QAS)
	to node[fill=white] {\cref{sec:auth-CP}}
	(HMCPPF);

\draw[->,thick]
	(HMCPPF)
	to node[fill=white] {\cref{sec:generic}}
	(SSLPF);

\draw[->,thick]
	(SSLPF)
	to node[fill=white] {{\cite{CMP20}}}
	(SSLCC);

\end{tikzpicture}
\end{center}
\caption{Relations between various notions considered in this work.}
\label{fig:links}
\end{figure}
}

\eav{
\begin{figure}
\begin{center}
\input{relations_small.tex}
\vspace{-1.5\baselineskip}
\end{center}
\caption{Relations between various notions considered in this work.}
\label{fig:links}
\end{figure}
}

In order to complete our main result, we show how to achieve honest-malicious copy protection for point functions, where the challenge distribution is
$(\Dhalf_p \times \Dhalf_p$), and correctness is also with respect to $\Dhalf_p$. To the best of our knowledge, this is the first  unconditional copy protection scheme; via the above reduction, it yields the first SSL scheme without assumptions.
See \Cref{fig:links} for a pictorial representation of the sequence of results. Our idea is to use a generic \emph{quantum message authentication scheme (QAS)} that satisfies the \emph{total authentication} property~\cite{GYZ17}. Briefly, a QAS is a private-key scheme with an encoding and decoding procedure such that the probability that the decoding accepts \emph{and} the output of the decoding in \emph{not} the original message is small. Security of a \emph{total} QAS is defined in terms of the existence of a \emph{simulator} that
reproduces the auxiliary register that an adversary has after attacking an encoded system, \emph{whenever} the verification accepts. An important feature of a total~QAS is that essentially no information about the key is leaked if the client accepts the authentication.

The main insight for the construction of honest-malicious copy protection for point functions from a total~QAS is to associate the key to the QAS with the point~$p$ in the point function. A copy protected program is thus an encoding of an arbitrary (but fixed) state~$\ket{\psi}$  into a total QAS, using~$p$ as the key. Given $p'$, the evaluation of the point function encoding is the QAS verification~\emph{with the key $p'$}. We thus get correctness in the case $p'=p$ from the correctness of the QAS; correctness in expectation for $p'\neq p$ follows with a bit more work. Importantly, the \emph{total} security property of the QAS gives us a handle on the auxiliary register that an adversary holds, \emph{in the case that the verification accepts}. Since Bob is honest, his evaluation  corresponds to the QAS verification map; in the case that Bob gets the challenge $x_1=p$, we  use the properties of the~total QAS to reason about Charlie's register, and we are able to show that Charlie's register cannot have much of a dependence on~$p$, which is to say that Charlie's outcome is necessarily independent of~$p$. This is sufficient to conclude that Bob and Charlie cannot win the copy protection game for a uniform point with probability much better than the trivial strategy in which Charlie makes an educated guess, given the challenge~$x_2$.
We note that total authentication is known to be satisfied by a scheme based on 2-designs~\cite{AM17}, as well as by the \emph{strong trap code}~\cite{DS18}.
Putting all of the above together, we obtain our main result, which is an explicit SSL scheme for point functions $P_p:\{0,1\}^n \to \{0,1\}$ which is $O(2^{-n})$-correct (on average) and $O(2^{-n})$-secure, under uniform sampling of~$p$ and where the challenge distribution is~$\Dhalf_p$.\footnote{This is achieved by instantiating the copy-protection scheme from \Cref{sec:auth-CP} with a total quantum authentication scheme given by \Cref{th:QAS-existence} and using it in the SSL construction of \Cref{sec:SSL}.} We note  the similarity between our approach for achieving honest-malicious copy protection and the approach in~\cite{Got03} in achieving tamper-evident encryption, based on quantum authentication codes.
We also mention a similarity with the blueprint in~\cite{CMP20}, which also produces a copy protected program starting from a private-key encryption scheme (in this case, the one of \cite{BL20}), associates a point with the key, and uses a type of verification of the integrity of the plaintext after decryption as the evaluation method.\looseness=-1

\fv{
\paragraph{Too good to be true?}
We emphasize that our results require  no assumptions at all, which is to say that the result is in the standard model (as opposed to, say  the QROM), and does not rely on any assumption on the computational power of the adversary. That either copy protection or SSL should be achievable in this model is very counter-intuitive, hence we explain here how we circumvent related impossibility results. In short, our work  strikes a delicate balance between correctness and security, in order to achieve the best of both worlds.

Prior work~\cite{Aar09} defines quantum copy protection assuming the adversary is given \emph{multiple identical} copies of the same copy protected state. Under this model, it is possible to show how an unbounded adversary can distinguish between the copy protected programs for different functions~\cite{Aar09}, which makes unconditionally secure copy protection impossible. In our scenario, we allow only a \emph{single} copy of the program state, hence this reasoning is not applicable.

Next, consider a scheme  (either copy protection or SSL) that is \emph{perfectly correct}, meaning that the outcome of the evaluation procedure is a deterministic bit. Clearly, such a scheme cannot be secure against unbounded adversaries, since \emph{in principle}, there is a sequence of measurements that an unbounded adversary can perform (via purification and rewinding), in order to perfectly obtain the truth table of the function. We conclude that perfectly correct schemes cannot satisfy our notion of unconditional security for copy protection.

We note that our scheme is, by design, not perfectly correct. This can be seen by reasoning about the properties of the~QAS: in any~QAS, it is necessary that, for a fixed encoding with key~$k$, there are a number of keys on which the verification accepts. The reason why this is true is similar to the argument above regarding perfect correctness: if this were not true, then the QAS (which is defined with respect to unbounded adversaries) would not be secure, since an adversary could in principle find~$k$ by trying all keys (coherently, so as to not disturb the quantum state) until one accepts. Somewhat paradoxically, it is this imperfection in the correctness that thus allows the unconditional security. Another way to understand the situation is that the honest evaluation in our copy protection (or SSL) scheme will unavoidably slightly damage the quantum encoding (even if performed coherently). In a brute-force attack, these errors necessarily accumulate to the point of rendering the program useless, and therefore the brute-force attack fails.

}

\fv{
%%%only appears in full version
%~~~~~~~~~~~~~~~~~~~~~~~~~~~~~~%
\subsection{Open Problems} %
\label{sec:open-problems}
%~~~~~~~~~~~~~~~~~~~~~~~~~~~~~~%
Our work leaves open a number of interesting avenues. For instance:
\begin{mylist}

\item \label{open:1}Could we show the more standard notion of correctness of our scheme, that is, correctness with respect to \emph{any} distribution?

\item  Is unconditional SSL achievable for a richer class of functions?

\item \label{open:3}Can our results on copy protection be extended to hold against \emph{two} malicious evaluators?

\end{mylist}
In \cref{sec:mm-correct}, we show  that \ref{open:1} and \ref{open:3} are related, by establishing that a point function copy protection scheme that is secure against two malicious evaluators and satisfies average correctness can be turned into a scheme that also satisfies the more standard notion of correctness.

%~~~~~~~~~~~~~~~~~~~~~~~~~~~~~~%
\subsection{Acknowledgements} %
%~~~~~~~~~~~~~~~~~~~~~~~~~~~~~~%
 We would like to thank Christian Majenz and Martti Karvonen for related discussions.
This material is based upon work supported by the Air Force Office of Scientific Research under award number FA9550-17-1-0083, Canada's   NFRF and NSERC, an Ontario ERA, and the University of Ottawa’s Research Chairs program.
SJ is a CIFAR Fellow in the Quantum Information Science program.

%~~~~~~~~~~~~~~~~~~~~~~~~~~~~~~%
\subsection{Outline} %
%~~~~~~~~~~~~~~~~~~~~~~~~~~~~~~%
The remainder of this document is structured as follows. In \Cref{sc:prelims}, we give background information on notation, basic notions and quantum message authentication. In \Cref{sec:definitions}, we define correctness and security for quantum copy protection and SSL. In \Cref{sec:generic}, we show the connection between malicious-malicious security, and standard correctness, as well as the links between honest-malicious copy protection and SSL. Finally, our main technical construction of honest-malicious copy protection from any total QAS is given in \Cref{sec:auth-CP}.
}

%=======================%
\section{Preliminaries} %
\label{sc:prelims}      %
%=======================%

%~~~~~~~~~~~~~~~~~~~~~%
\subsection{Notation} %
\label{sc:notation}   %
%~~~~~~~~~~~~~~~~~~~~~%

All Hilbert spaces are complex and of finite dimensions. We usually
denote a Hilbert space using a sans-serif font such as $\tsf{S}$ or
$\tsf{H}$. We will often omit writing the tensor symbol when taking the
tensor product of two Hilbert spaces, i.e.:
$\tsf{A} \tensor \tsf{B} = \tsf{AB}$. We use the Dirac
notation~\cite{Dir39} throughout, which is to say that
$\ket{\psi} \in \tsf{H}$ denotes a unit vector and
$\bra{\psi} : \tsf{H} \to \C$ denotes the corresponding linear map in
the dual space. Finally, Hilbert spaces may be referred to as
``registers'', acknowledging that they sometimes model physical objects
which may be sent, kept, discarded, etc., by parties participating in
quantum information processing tasks.

The set of linear operators, unitary operators, and density operators on
a Hilbert space $\tsf{H}$ are denoted by $\mc{L}(\tsf{H})$,
$\mc{U}(\tsf{H})$, and $\mc{D}(\tsf{H})$ respectively. A linear operator
will often be accompanied by a subscript indicating the Hilbert space on
which it acts. This will be useful for bookkeeping and to occasionally
omit superfluous identities. For example, if
$L_\tsf{A} \in \mc{L}(\tsf{A})$ and $\ket{\psi}_\tsf{AB} \in \tsf{AB}$,
then
\begin{equation}
	L_\tsf{A} \ket{\psi}_\tsf{AB}
	=
	\left(L_\tsf{A} \tensor I_\tsf{B}\right)
	\ket{\psi}_\tsf{AB}.
\end{equation}

For a function $f : X \to \C$, for some finite set $X$, when no distribution on $x$ is clear from context, we write
\begin{equation}
	\E_x f(x) = \frac{1}{\abs{X}} \sum_{x \in X} f(x).
\end{equation}
In other words, when there is no implicit distribution associated with $x$, we write $\E_x f(x)$ to denote the expectation of $f(x)$
if $x$ is sampled uniformly at random from the domain of $f$.

For a distribution $D$ on a set $S$, we will use the notation $x\leftarrow D$ to denote that variable $x$ is sampled from $D$, and $D(x)$ to denote the probability that a given $x\in S$ is sampled.

Throughout this work, we will denote a family of Boolean circuits on $n$ bits as $\calC$. The circuit families of specific interest in this work are point functions and compute-and-compare-functions. These are defined below.

Let $n \in \N$ and $p \in \bool^n$. A point function $P_p$ takes as input $x \in \bool^n$ and is defined as:
\begin{equation}
	P_p(x) = \begin{cases}
	1 & \text{if } x = p, \\
	0 & \text{otherwise.}
	\end{cases}
\end{equation}

A closely related but more general class of circuits than point functions are compute-and-compare circuits ($\CC$). Formally, for a function $f : \bool^n \rightarrow \bool^m$ and $y \in \bool^m$ in its range, the corresponding compute-and-compare function $\CC_y^f$ takes $x \in \bool^n$ as input and is defined as:
\begin{equation}
	\CC_y^f(x) = \begin{cases}
	1 & \text{if } f(x) = y, \\
	0 & \text{otherwise.}
	\end{cases}
\end{equation}
Clearly, when $f$ is the identity map, we recover point functions under this definition.

%~~~~~~~~~~~~~~~~~~~~~~~~~~~~~~~~~~~~~~~~~~~~~~%
\subsection{Pairwise Independent Permutations} %
\label{sc:pairwise}                            %
%~~~~~~~~~~~~~~~~~~~~~~~~~~~~~~~~~~~~~~~~~~~~~~%

The notion of pairwise independent hash functions, first defined by
Carter and Wegman \cite{WC81} under the name of strongly universal$_2$
functions, is a commonly used tool in cryptography. Essentially, a
pairwise independent hash function is a family of functions
$\{h_r : A \to B\}_{r\in{\cal R}}$ which behaves like the set of all
functions $\{h : A \to B\}$ if we are limited to only observing two
input-output pairs from these functions.

A closely related notion is the idea of a pairwise independent
permutation (\textit{e.g.}: \cite{NR99}), which we recall below.

\begin{definition}[Pairwise Inependent Permutation]
A pairwise independent permutation on $\{0,1\}^n$ is a
family of functions $\left\{h_r : \{0,1\}^n \to \{0,1\}^n\right\}_{r \in {\cal R}}$ for some finite set ${\cal R}$ such that
\begin{enumerate}
	\item
		every $h_r$ is a permutation and
	\item
		for all distinct $x_0, x_1 \in \{0,1\}^n$ and distinct $y_0,y_1\in\{0,1\}^n$,
		\begin{equation}
		\label{eq:pairwise-TVD}
				\Pr_{r}\left[
					(h_r(x_0), h_r(x_1)) = (y_0,y_1)
				\right]
				=
				\frac{1}{2^n}\frac{1}{2^n-1}
		\end{equation}
		where $r$ is sampled uniformly at random from ${\cal R}$.
\end{enumerate}
\end{definition}

A straightforward construction of a pairwise independent permutation on
bit strings $\{0,1\}^n$, mentioned in \cite{NR99}, is to consider all
functions in the finite field of $2^n$ elements of the form
$x \mapsto m \cdot x + b$ where $r = (m, b)$ and $m$ is not the zero
element.

%~~~~~~~~~~~~~~~~~~~~~~~~~~~%
\subsection{Trace Distance} %
\label{sc:trace-distance}   %
%~~~~~~~~~~~~~~~~~~~~~~~~~~~%

We recall the definition of the trace norm and the trace distance
between linear operators. Our definitions are taken from \cite{Wat18}.

\begin{definition}
\label{df:trace-norm}
Let $\tsf{A}$ be a Hilbert space. For any linear operator
$X \in \mc{L}(\tsf{A})$, we define the \emph{trace norm} of $X$ as
\begin{equation}
\label{eq:trace_norm}
	\norm{X}_1
	=
	\max_{U \in \mc{U}(\tsf{A})}\abs{\ip{U}{X}}
\end{equation}
where $\ip{U}{X} = \Tr\left[U^\dag X\right]$. We include the subscript
$1$ to recall that this is the Schatten-$1$ norm.
\end{definition}

From the trace norm, we can now define the trace distance.

\begin{definition}
Let $\tsf{A}$ be a Hilbert space. The \emph{trace distance} between any
two linear operators on this space $X, Y \in \mc{L}(\tsf{A})$ is given
by
\begin{equation}
	\Delta(X, Y)
	=
	\frac{1}{2}\norm{X - Y}_1.
\end{equation}
If $\Delta(X, Y) \leq \epsilon$, we may write $X \approx_\epsilon Y$.
\end{definition}

Note that our definition of the trace distance differs from the one
offered in \cite{Wat18} by including a factor of $\frac{1}{2}$. This
factor is common in quantum information (e.g.: \cite{NC00}) as it
ensures that the trace distance between two density operators $\rho$ and
$\sigma$, i.e.: states of quantum systems, is in the interval $[0,1]$.

Finally, we give a technical lemma pertaining to the trace distance
between two bipartite states written as mixtures where the state on one
of the subsystems is always given by a pure state taken from some set of
orthogonal states. For completeness, a proof of this lemma is included
in \cref{sc:prelim-proofs}.

\begin{lemma}
\label{th:trace-distance-orthogonal}
Let $\tsf{A}$ and $\tsf{B}$ be Hilbert spaces. Let
$\{\ket{\psi_j}\}_{j \in J} \subseteq \tsf{A}$ be a collection of
orthogonal states. Then, for any collections of linear operators
$\{X_j\}_{j \in J}$ and $\{Y_j\}_{j \in J}$ on $\tsf{B}$, we have that
\begin{equation}
	\Delta\left(
		\sum_{j \in J} \ketbra{\psi_j} \tensor X_j,
		\sum_{j \in J} \ketbra{\psi_j} \tensor Y_j
	\right)
	=
	\sum_{j \in J} \Delta\left(
		X_j,
		Y_j
	\right)
	.
\end{equation}
\end{lemma}

Noting that for any scalar $s$ we have that
$\Delta(s \cdot X, s \cdot Y) = \abs{s} \cdot \Delta(X, Y)$, we obtain
as a direct corollary to the above lemma that
\begin{equation}
	\Delta\left(
		\E_{j \in J} \ketbra{\psi_j} \tensor X_j,
		\E_{j \in J} \ketbra{\psi_j} \tensor Y_j
	\right)
	=
	\E_{j \in J} \Delta\left(
		X_j,
		Y_j
	\right)
	.
\end{equation}

%~~~~~~~~~~~~~~~~~~~~~~~~~~~~~~~~~~~%
\subsection{Quantum Authentication} %
\label{sc:qas}                      %
%~~~~~~~~~~~~~~~~~~~~~~~~~~~~~~~~~~~%

We recall the definition of total quantum authentication from
\cite{GYZ17} and highlight a few properties of such schemes.

\begin{definition}
\label{df:QAS}
An authentication scheme $\tsf{QAS}$ for the Hilbert space $\tsf{M}$ is
a pair of keyed CPTP maps
\begin{equation}
\label{eq:QAS-syntax}
	\tsf{QAS.Auth}_k : \mc{L}(\tsf{M}) \to \mc{L}(\tsf{Y})
	\qq{and}
	\tsf{QAS.Ver}_k : \mc{L}(\tsf{Y}) \to \mc{L}(\tsf{MF})
	\qq{for keys}
	k \in \mc{K}
\end{equation}
and where $\tsf{F}$ admits $\{\ket{\text{Acc}}, \ket{\text{Rej}}\}$ as
an orthonormal basis. Moreover, these maps are such that for all
states $\rho \in \mc{D}(\tsf{M})$ and all keys $k \in \mc{K}$ we have
that
\begin{equation}
\label{eq:QAS-correctness}
	\tsf{QAS.Ver}_k \circ \tsf{QAS.Auth}_k(\rho)
	=
	\rho \tensor \ketbra{\text{Acc}}.
\end{equation}
\end{definition}

To facilitate our analysis, we will make the same simplifying
assumptions as in \cite{GYZ17} on any quantum authentication scheme
considered in this work.
\begin{enumerate}
	\item
		We assume that $\tsf{QAS.Auth}_k$ can be modeled by an isometry.
		Specifically, we assume that
		\begin{equation}
		\label{eq:QAS-assumption-iso}
			\tsf{QAS.Auth}_k(\rho) = A_k \rho A_k^\dag
		\end{equation}
		for some isometry $A_k \in \mc{L}(\tsf{M}, \tsf{Y})$.
	\item
		For all keys $k \in \mc{K}$, as $A_k$ is an isometry, $A_kA_k^\dagger$ is the
		projector onto the image of $A_k$. In other words, it projects
		onto valid authenticated states for the key $k$. We then assume
		that $\tsf{QAS.Ver}_k$ is given by the map
		\begin{equation}
		\label{eq:QAS-assumption-ver}
			\rho
			\mapsto
			A_k^\dag  \rho  A_k
			\tensor
			\ketbra{\text{Acc}}
			+
			\Tr\left[\left(I - A_kA_k^\dag\right)\rho\right]
			\cdot
			\frac{I}{\dim(\tsf{M})} \tensor \ketbra{\text{Rej}}.
		\end{equation}
		In other words, $\tsf{QAS.Ver}_k$ verifies if the state is a
		valid encoded state. If it is, then it inverts the
		authentication procedure and adds an ``accept'' flag. If it is
		not, then it outputs the maximally mixed state and adds a
		``reject'' flag.
\end{enumerate}

Finally, we will also define the map $\tsf{QAS.Ver}'_k$ by
\begin{equation}
	\rho
	\mapsto
	\left(I_\reg{M} \tensor \bra{\text{Acc}}_\reg{F}\right)
	\tsf{QAS.Ver}_k(\rho)
	\left(I_\reg{M} \tensor \ket{\text{Acc}}_\reg{F}\right)
	=
	A_k^\dag\rho A_k
	.
\end{equation}
Essentially, this map outputs a subnormalized state corresponding to the
state of the message register $\reg{M}$ conditioned on the verification
procedure accepting the state. In particular, note that the probability
that the verification procedure accepts the state $\rho$ when using the
key $k$ is given by $\Tr\left(\tsf{QAS.Ver}'_k(\rho)\right)$.

\Cref{df:QAS} does not make any type of security guarantee on an
authentication scheme. It only specifies a syntax, \cref{eq:QAS-syntax},
and a correctness guarantee, \cref{eq:QAS-correctness}. The following
definition describes the security notion of $\epsilon$-total
authentication. Note that this security definition differs from some
early notions of security for quantum authentication schemes
\cite{BCG+02,DNS12}.

\begin{definition}
\label{df:QAS-security}
An authentication scheme $\tsf{QAS}$ is an $\epsilon$-total
authentication scheme if for all CPTP maps
$\Phi : \mc{L}(\tsf{YZ}) \to \mc{L}(\tsf{YZ})$ there exists a completely
positive trace non-increasing map
$\Psi : \mc{L}(\tsf{Z}) \to \mc{L}(\tsf{Z})$ such that
\begin{equation}
\label{eq:total-authentication}
	\E_{k \in \mc{K}}
		\ketbra{k}
		\tensor
		\tsf{QAS.Ver}'_k \circ \Phi \circ \tsf{QAS.Auth}_k(\rho)
	\approx_\epsilon
	\E_{k \in \mc{K}}
		\ketbra{k}
		\tensor
		\tsf{QAS.Ver}'_k \circ \Psi \circ \tsf{QAS.Auth}_k(\rho)
\end{equation}
for any state $\rho \in \mc{D}(\tsf{MZ})$.
\end{definition}

A key difference between the \cite{GYZ17} security notion of
authentication and previous notions is the explicit $\ketbra{k}$ state
which appears in \cref{eq:total-authentication}. The existence of this
key register will be used, with the help of
\cref{th:trace-distance-orthogonal}, in some of our technical arguments,
such as the proof of \cref{lem:Bob-correct}.

Note that our discussion, unlike the one in \cite{GYZ17}, omits
adding another register $\tsf{S}$ to model all other information that
a sender and receiver could share as part of a larger protocol but which
is not directly implicated in the authentication scheme. Such a register
is not needed in our analysis.

Next, we give a lemma which upper bounds the probability that any
fixed state is accepted by the verification procedure, when averaged
over all possible keys. For completeness, the proof can be found in
\cref{sc:prelim-proofs}.

\begin{lemma}
\label{th:QAS-wrong-key}
Let $\tsf{QAS}$ be an $\epsilon$-total authentication scheme on the
Hilbert space $\tsf{M}$ of dimension greater or equal to $2$. Then, for
any $\rho \in \mc{D}(\reg{Y})$, we have that
\begin{equation}
	\E_{k \in \mc{K}}
	\Tr\left[
		\tsf{QAS.Ver}'_{k}(\rho)
	\right]
	\leq
	2\epsilon.
\end{equation}
\end{lemma}

Finally, we give an existence lemma. The proof is also given in
\cref{sc:prelim-proofs}. It essentially follows from a theorem
describing how unitary $2$-designs (as introduced in \cite{DCEL09}) can
be used to construct total quantum authentication schemes \cite{AM17}
and then choosing a suitable unitary $2$-design \cite{CLLW16}. A
few additional technical arguments are needed to ensure that the key set is
precisely the bit strings of a given length.

\begin{lemma}
\label{th:QAS-existence}
For any strictly positive integers $n$ and $k$, there exists a
$\left(5 \cdot 2^{\frac{5n-k}{16}}\right)$-total quantum authentication
scheme on $n$ qubits with key set $\{0,1\}^k$.
\end{lemma}

%=======================%
\section{Definitions}   %
\label{sec:definitions} %
%=======================%

Here, we define quantum copy protection (\Cref{sec:defn-copy-protection}) and secure software leasing (\Cref{sec:defn-SSL}), along with their correctness and security notions. All of our definitions are for Boolean circuits only, where the input is a binary string, and the output is a single bit.

%~~~~~~~~~~~~~~~~~~~~~~~~~~~~~~~~~~~~%
\subsection{Quantum Copy Protection} %
\label{sec:defn-copy-protection}     %
%~~~~~~~~~~~~~~~~~~~~~~~~~~~~~~~~~~~~%

We present our definition of a copy protection scheme, following the general lines of~\cite{CMP20}. We note that we have rephrased the definition in~\cite{CMP20} in terms of the more standard cryptographic notion where the parameter in the definition (here, we use $\epsilon$) characterizes the \emph{insecurity} of a game (and hence, we strive for schemes were $\epsilon$ is small).

%--------------------------------------------------%
\subsubsection{Quantum Copy Protection Scheme} %
%--------------------------------------------------%

First, we define the functionality of \emph{quantum copy protection}.
\begin{definition}[Quantum copy protection scheme] Let $\calC$ be a set of $n$-bit Boolean circuits. A \emph{quantum copy protection} scheme for $\mathcal{C}$ is a pair of quantum circuits $\CP=(\Protect,\Eval)$ such that for some space $\sf Y$:
\begin{enumerate}
\item $\Protect(C)$: takes as input a Boolean circuit $C \in \calC$, and outputs a  quantum state $\rho\in {\cal D}({\sf Y})$.
\item  $\Eval(\rho, x)$: takes a quantum state $\rho\in {\cal D}({\sf Y})$ and string $x \in \bool^\len$ as inputs and outputs a bit $b$.
\end{enumerate}
\end{definition}

We will interpret the output of $\Protect$ and $\Eval$ as quantum states on $\sf Y$ and $\mathbb{C}^2$, respectively, so that, for example, for any bit $b$, string $x$ and program $\rho$, $\Tr[\ketbra{b}\Eval(\rho, x)]$ is the probability that $\Eval(\rho,x)$ outputs $b$.
\begin{definition}[$\eta$-Correctness of copy protection]
\label{def:cp-correct}
A \emph{quantum copy protection} scheme for a set of $n$-bit circuits $\calC$, $\CP$, is \emph{$\eta$-correct} with respect to a family of distributions on $n$-bit strings $\{T_C\}_{C \in \calC}$, if for any $C \in \calC$ and $\rho = \Protect(C)$, the scheme satisfies
\begin{equation}
\E_{x \leftarrow T_C} \Tr[\ketbra{C(x)}\Eval(\rho, x)] \geq 1-\eta.
\end{equation}
\end{definition}

Our notion of correctness differs from that of \cite{CMP20}, and other previous work on uncloneable point function obfuscation, by being defined with respect to a family of distributions (see \Cref{sec:open-problems}). However, if the scheme is $\eta$-correct with respect to all families of distributions, then we recover the more standard definition of correctness.

\subsubsection{Reusability}\label{sec:reuse}

We note that the $\Eval$ procedure only addresses the ability to compute $C$ on a single input~$x$. Thus, the \emph{reusability} of $\rho$ is not addressed in the definition. However, some notion of reusability follows from correctness. Let $W_x$ be a unitary purification of $\Eval$ on $\reg{YZO}$, where $\reg{ZO}$ is the purifying space, and we assume the output qubit is the single-qubit register $\reg{O}$. For a single qubit space $\reg{O}'$, define:
\begin{align*}
\overline{\Eval}(\rho,x) = \Tr_{\reg{ZO}}\left(W_x^\dagger\cdot {\sf CNOT}_{\reg{OO}'}\cdot W_x(\rho\otimes \braket{0}_{\reg{Z}}\otimes\braket{0}_{\reg{O}}\otimes\braket{0}_{\reg{O}'})W_x^\dagger\cdot {\sf CNOT}_{\reg{OO}'}\cdot W_x\right),
\end{align*}
as shown in \cref{fig:pure-eval}.

\begin{figure}
\centering
\begin{tikzpicture}[scale=1.2]
\node at (-.25,1) {$x$}; 		\draw[double] (0,1)--(1.5,1)--(1.5,.75);
\node at (-.25,.5) {$\rho$};	\draw (0,.5) -- (4,.5);
	\node at (.5,0) {$\ket{0}$};		\draw (.75,0)--(3.5,0);
	\node at (.5,-.5) {$\ket{0}$};		\draw (.75,-.5)--(3.5,-.5);
	\node at (.5,-1) {$\ket{0}$};		\draw (.75,-1)--(4,-1);

\node at (1,.65) {\small $\reg{Y}$};
\node at (1,.15) {\small $\reg{Z}$};
\node at (1,-.35) {\small $\reg{O}$};
\node at (1,-.85) {\small $\reg{O}'$};

\filldraw[fill=white] (1.25,.75) rectangle (1.75,-.75);
\node at (1.5,0) {$W_x$};

\filldraw (2.125,-.5) circle (.035);
\draw (2.125,-.5)--(2.125,-1.1);
\draw (2.125,-1) circle (.1);

\filldraw[fill=white] (2.5,.75) rectangle (3,-.75);
\node at (2.75,0) {$W_x^\dagger$};

\draw[dashed] (.25,1.25) rectangle (3.75,-1.25);
\node at (3.1,1.06) {$\overline{\Eval}$};
\end{tikzpicture}
\caption{An evaluation procedure that outputs the program for reuse.}\label{fig:pure-eval}
\end{figure}
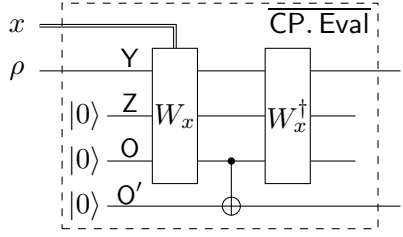

If $x\leftarrow T_C$, and $\CP$ is $\eta$-correct with respect to $T_C$, then the post-evaluated state, $\tilde\rho$ output by $\overline{\Eval}(\rho,x)$ on register $\sf Y$ satisfies
$\Delta(\rho,\tilde\rho)\leq O(\eta)$.

\subsubsection{Honest-Malicious Security for Quantum Copy Protection}
\label{sec:definition-CP}

In this section, we define the notion of security for a copy protection scheme against an adversary $\adv = (\pirate, \adv_1, \adv_{2})$, where $\cal P$ (Pete) is the \emph{pirate}, and $\adv_1$ (Bob) and $\adv_2$ (Charlie) are \emph{users} (see \cref{fig:freeloading-experiment}).
We use the $\ExperimentFree$ from~\cite[Section 3]{CMP20} as the basis
of our security game between a challenger and~$\adv$. The game is
parametrized by: (i) a distribution $D$ on the set of circuits $\calC$,
and (ii) a set of distributions $\{D_C\}_{C \in \calC}$ over pairs of
input strings in $\{0,1\}^n\times\{0,1\}^n$, called the \emph{challenge
distributions}.

\vspace{0.5em}
\begin{center}
\fbox{
\begin{minipage}[c]{0.8\textwidth}
\textbf{The CP experiment $\ExperimentFree$}$_{\adv,\CP}$
\begin{enumerate}
\item The challenger samples $C \leftarrow D$ and sends $\rho = \Protect(C)$ to $\pirate$.
\item $\pirate$ outputs a state $\sigma$ on registers $\reg{A}_1, \reg{A}_{2}$ and sends $\reg{A}_1$ to $\adv_1$ and $\reg{A}_2$ to $\adv_2$.
\item At this point, $\adv_1$ and $\adv_{2}$ are separated and cannot communicate.  The challenger samples $(x_1, x_{2}) \leftarrow D_C$ and sends $x_1$ to $\adv_1$ and $x_2$ to $\adv_2$.
    \item $\adv_1$ returns a bit $b_1$ to the challenger and $\adv_2$ returns a bit $b_2$.
    \item The challenger outputs $1$ if and only if $b_1 = C(x_1)$ and
	$b_2 = C(x_2)$, in which case, we say that $\adv$ wins the game.
\end{enumerate}
\end{minipage}
}
\end{center}
\vspace{0.5em}

\begin{figure}
\centering
\begin{tikzpicture}
\draw[double] (0,1.25)--(2.675,1.25)--(2.675,.75);
	\draw (1.5,.55)--(3,.55);	\draw[double] (3,.55)--(3.75,.55);
\draw (0,0)--(1.5,0);
	\draw (1.5,-.55)--(3,-.55);	\draw[double] (3,-.55)--(3.75,-.55);
\draw[double] (0,-1.25)--(2.675,-1.25)--(2.675,-.75);

\node at (-3.5,0) {
	$\begin{aligned}
		C & \leftarrow D\\
		\rho & = \Protect(C)\\
		(x_1,x_2) & \leftarrow D_C
	\end{aligned}$
};

\node at (-.25,1.25) {$x_1$};
\node at (-.25,0) {$\rho$};
\node at (-.25,-1.25) {$x_2$};
\node at (.25,.15) {\small $\sf Y$};

\filldraw[fill=white] (.5,1) rectangle (1.5,-1);
\node at (1,0) {
	\begin{minipage}{2cm}
		\centering
		${\cal P}$\\
		\small (Pete)
	\end{minipage}
};
\node at (1.75,.75) {\small $\reg{A}_1$};
\node at (1.75,-.35) {\small $\reg{A}_2$};

\filldraw[fill=white] (2,1) rectangle (3.5,.1);
\node at (2.75,.55) {\begin{minipage}{2cm}
	\centering\small
	$\adv_1$\\
	(Bob)
\end{minipage}};
\filldraw[fill=white] (2,-1) rectangle (3.5,-.1);
\node at (2.75,-.55) {\begin{minipage}{2cm}
	\centering\small
	$\adv_2$\\
	(Charlie)
\end{minipage}};

\node at (4,.5) {$b_1$};
\node at (4,-.5) {$b_2$};

\node at (6.5,0) {
\begin{minipage}{3.75cm}
\centering
\underline{Winning Conditions:}\\
\vskip 5pt
$b_1=C(x_1)$\\
$b_2=C(x_2)$
\end{minipage}
};
\end{tikzpicture}
\caption{The pirating game $\ExperimentFree_{\adv,\CP}$}\label{fig:freeloading-experiment}
\end{figure}

\vspace{1em}
In previous work on copy protection, the adversary is assumed to control ${\cal P}$, ${\cal A}_1$ and ${\cal A}_2$, whose behaviour can be arbitrary (or, in some cases, computationally bounded). This models a setting where the potential users of pirated software are aware that the software is pirated, and willing to run their software in some non-standard way in order to make use of it. We refer to this setting as the \emph{malicious-malicious} setting. In this setting, the action of the adversary $\adv=(\pirate,\adv_1,\adv_2)$ can be specified by:
\begin{enumerate}
\item an arbitrary CPTP map $\Phi_{\cal P}:{\cal L}(\sf Y)\to{\cal L}(\reg{A}_1\reg{A}_2)$, representing the action of $\cal P$, where $\reg{A}_1$ and $\reg{A}_2$ are arbitrary spaces;
\item arbitrary two-outcome projective measurements $\{\Pi_x\}_{x\in\{0,1\}^n}$ on $\reg{A}_1$, such that $\adv_1$ (Bob) performs the measurement $\{\Pi_{x_1},I-\Pi_{x_1}\}$ on input $x_1$ to obtain his output bit $b_1$; and
\item arbitrary two-outcome projective measurements $\{\Pi_x'\}_{x\in\{0,1\}^n}$ on $\reg{A}_2$, such that $\adv_2$ (Charlie) performs the measurement $\{\Pi_{x_2}',I-\Pi_{x_2}'\}$ on input $x_2$ to obtain his output bit $b_2$.
\end{enumerate}

In contrast, one could also imagine a scenario in which users are honest, and will therefore try to evalute the program they receive from ${\cal P}$ by running $\Eval$. In that case, while $\pirate$ can still perform an arbitrary CPTP map, $\adv_1$ and $\adv_2$ are constrained to run $\Eval$.
It is potentially easier to design copy protection in this weaker setting, which we call the \emph{honest-honest} setting, since the adversary is more constrained. We will consider an intermediate setting.

Diverging from previous work, we will focus on a special type of adversary, where $\adv_1$ (Bob) performs the \emph{honest} evaluation procedure, while $\adv_{2}$ (Charlie) performs an arbitrary measurement. (See \Cref{sec:intro-contributions} for a discussion of this model). Specifically, we consider the following type of adversary.

\begin{definition}
An \emph{honest-malicious adversary} for the pirating game is an adversary of the form $\hatAdv = (\pirate, \Eval, \adv_2)$, where $\pirate$ implements an arbitrary CPTP map $\Phi_{\pirate}:{\cal L}(\sf Y)\to {\cal L}({\sf YA}_2)$, ${\sf A}_2$ is any space, and $\adv_2$ is specified by a set of arbitrary two-outcome measurements $\{\Pi_x\}_{x\in\{0,1\}^n}$ on~${\sf A}_2$.
\end{definition}

For a fixed scheme $\CP = (\Protect, \Eval)$ for a set of $n$-bit circuits $\calC$, we define \emph{honest-malicious} security with respect to distributions $D$ and $\{D_C\}_{C \in \calC}$ in terms of the best possible winning probability, $\Pr[\ExperimentFree_{\hatAdv,\CP}]$, over honest-malicious adversaries~$\hatAdv$.
Observe that there is one strategy that $\pirate$ can always facilitate, which is to pass the  intact program to Bob and then let Charlie locally produce his best guess of the output, based on prior knowledge of $D$ and $\{D_C\}_{C \in \calC}$\footnote{There are other trivial strategies, $\emph{e.g.},$ where Charlie gets an intact program register and Bob does not, but this is a more restricted trivial strategy, since Bob is constrained to evaluate the program honestly.}. This leads to a winning probability for the above game which is truly trivial to achieve, in the sense that Charlie is using a strategy that does not take any advantage of the interaction with the pirate $\pirate$. In fact, assuming the scheme is $\eta$-correct with respect to the distribution family $\{T_C\}_{C\in {\cal C}}$ where $T_C$ is Bob's marginal of $D_C$,
Bob will always produce the correct answer, except with probability
$\eta$. Indeed, Charlie simply considers the most likely output, given
his input, thereby upper bounding the winning probability with Charlie's
maximum guessing probability\footnote{The winning probability may be
less than this. By the union bound, even though Bob's and Charlie's
inputs are not independent, the overall success probability will be at
least $p^{\text{marg}}-\eta$, and we will be considering situations
where $\eta$ is small.}.

Formally, we define $\pmarg_{D, \{D_C\}_{C \in \calC}}$ as follows.
The distributions $D$ and $\{D_C\}_{C \in \mc{C}}$ yield a joint
distribution $\tilde{D}$ on $\mc{C} \times \{0,1\}^n$ by first sampling
$C \gets D$ and then sampling $(x_1, x_2) \gets D_C$ and only taking the
$x_2$ component. Let $\hat{D}$ be the marginal distribution of $x_2$
from $\tilde{D}$ and, for every $x$, let $\hat{D}_x$ be the marginal
distribution of $C$ from $\tilde{D}$, conditioned on $x_2 = x$. Then,
\begin{equation}
\label{eq:pmarg}
 \pmarg_{ D, \{D_C\}_{C \in \calC}}
 =
\E_{x \gets \hat{D}} \max_{b \in \{0,1\}} \Pr_{C \gets \hat{D}_x}[C(x) = b].
\end{equation}

This is different from the security notion in~\cite{CMP20} where the trivial guessing probability is optimized over both users. For intuition, note that $\pmarg$ is always at least $1/2$, since Charlie can always output a random bit that is correct with probability $1/2$. Depending on the specific input and challenge distributions, it may be larger. We now state the main security notion for this work.

\begin{definition}[Honest-malicious security]
\label{defn:copy-protection}
A copy protection scheme $\CP = (\Protect, \Eval)$ for a set of $n$-bit circuits $\calC$ is \emph{$\epsilon$-honest-malicious secure with respect to the distribution $D$ and challenge distributions $\{D_C\}_{C \in \calC}$} if for all honest-malicious adversaries $\hatAdv$,
\begin{equation}
 \Pr[\ExperimentFree_{\hatAdv,\CP}] \leq \pmarg + \epsilon\,,
\end{equation}
where  $\pmarg = \pmarg_{D, \{D_C\}_{C \in \calC}}$.
\end{definition}

We re-iterate that our definition for honest-malicious security is \emph{statistical}: it makes no assumption on the computational power of~$\hatAdv$ (see \Cref{sec:intro-contributions}).

Finally, if we modify the above definition by allowing arbitrary adversaries $\adv=(\pirate,\adv_1,\adv_2)$, and letting $\bar{p}^{\text{marg}}$ denote the optimal trivial guessing probability, as in \cref{eq:pmarg} but over \emph{both} adversaries (see also \cite{CMP20}), we recover the more standard security definition, which we call \emph{malicious-malicious} security:

\begin{definition}[Malicious-malicious security]
\label{defn:mm-copy-protection}
A copy protection scheme $\CP$ for a set of $n$-bit circuits $\calC$ is \emph{$\epsilon$-malicious-malicious secure with respect to the distribution $D$ and challenge distributions $\{D_C\}_{C \in \calC}$} if for \emph{all} adversaries $\adv$,
\begin{equation}
 \Pr[\ExperimentFree_{\adv,\CP}] \leq \bar{p}^{\text{marg}} + \epsilon.
\end{equation}
\end{definition}

%~~~~~~~~~~~~~~~~~~~~~~~~~~~~~~~~~~~~%
\subsection{Secure Software Leasing} %
\label{sec:defn-SSL}                 %
%~~~~~~~~~~~~~~~~~~~~~~~~~~~~~~~~~~~~%

We define Secure Software Leasing (SSL) below. As with copy protection, the basic scheme and security game mirror~\cite{CMP20} but we diverge from them in our exact notions of correctness and security.

%--------------------------------------------------%
\subsubsection{Secure Software Leasing Scheme} %
\label{sec:ssl-definition}                         %
%--------------------------------------------------%

\begin{definition}[Secure software leasing (SSL)] Let $\mathcal{C}$ be a set of $n$-bit Boolean circuits. A \emph{secure software leasing} scheme for $\mathcal{C}$ is a tuple of quantum circuits $\SSL=(\SSLGen,$ $\SSLLease$, $\SSLEval$, $\SSLVerify)$ such that for some space $\sf Y$:
\begin{enumerate}
\item $\SSLGen$: outputs a secret key $\sk$.
\item $\SSLLease(\sk, C)$: takes as input a secret key $\sk$ and a circuit $C \in \calC$, and outputs a quantum state $\rho\in{\cal D}({\sf Y})$.
\item $\SSLEval(\rho, x)$: takes as input a quantum state $\rho\in{\cal D}({\sf Y})$ and input string $x\in\{0,1\}^n$ and outputs a bit~$b$ together with a post-evaluated state $\tilde{\rho}\in{\cal D}({\sf Y})$.
\item $\SSLVerify(\sk, \sigma, C)$: takes a secret key $\sk$, a circuit $C \in \calC$ and a quantum state $\sigma\in{\cal D}({\sf Y})$,
and outputs a bit $v$ indicating acceptance or rejection.
\end{enumerate}
\end{definition}

\begin{definition}[$\eta$-Correctness of SSL]
\label{def:ssl-correct}
A \emph{secure software leasing} scheme for $\calC$, $\SSL$, is $\eta$-\emph{correct} with respect to a family of distributions on $n$-bit strings $\{T_C\}_{C \in \calC}$, if for any $C \in \calC$, $\sk\leftarrow\SSLGen$, and $\rho = \SSLLease(\sk, C)$, the scheme satisfies:
\begin{itemize}
\item Correctness of Evaluation:
$\displaystyle \E_{x \leftarrow T_C} \Tr\left(\ketbra{C(x)}\SSLEval(\rho, x)]\right) \geq 1-\eta$,
\item and Correctness of Verification:
$\displaystyle\Tr\left(\ketbra{1}\SSLVerify(\sk, \rho, C)\right) \geq 1-\eta$.
\end{itemize}
\end{definition}
In the above definition, recall that for $b\in\{0,1\}$, $\Tr\left(\ketbra{b}\SSLVerify(\sk,\rho, C)\right)$ is the probability that $\SSLVerify(\sk,\rho,C)$ outputs the bit $b$, and similarly for $\Tr\left(\ketbra{b}\SSLEval(\rho, x)\right)$.

When a scheme $\SSL$ is $\eta$-correct with respect to every distribution, we recover the more standard notion of correctness.

About the definition of correctness for SSL, we remark that as stated, it seems to only imply that the lessee can either run the program, \emph{or} return it. The definition does not explicitly guarantee that the post-evaluated state output by $\SSLEval$ after the program has been run will be accepted by $\SSLVerify$. However, using the construction described in \cref{sec:reuse}, it is always possible to evaluate the program, and by correctness of evaluation, the program will not be changed very much, and so by correctness of verification, it will still be accepted with high probability. The probability of acceptance will degrade by $O(\eta)$ with each evaluation.

%----------------------------------------------------%
\subsubsection{Security for Secure Software Leasing} %
\label{sec:security-SSL}                             %
%----------------------------------------------------%

\begin{figure}
\centering
\begin{tikzpicture}
	\draw (1.5,.55)--(3,.55);	\draw[double] (4.5,.55)--(5.5,.55);
\draw (0,0)--(1.5,0);
	\draw (1.5,-.55)--(3,-.55);	\draw[double] (3,-.55)--(5.5,-.55);
\draw[double] (0,-1.25)--(3.25,-1.25)--(3.25,-.75);

\node at (-3.25,0) {
	$\begin{aligned}
		\sk & \leftarrow \SSLGen\\
		C & \leftarrow D\\
		\rho & = \SSLLease(\sk,C)\\
		x & \leftarrow D_C'
	\end{aligned}$
};

\node at (-.25,0) {$\rho$};
\node at (-.25,-1.25) {$x$};
\node at (.25,.15) {\small $\sf Y$};

\filldraw[fill=white] (.5,1) rectangle (1.5,-1);
\node at (1,0) {
	\begin{minipage}{2cm}
		\centering
		$\Phi_{\adv}$
	\end{minipage}
};
\node at (1.75,.75) {\small $\reg{Y}$};
\node at (1.75,-.35) {\small $\reg{A}$};

\filldraw[fill=white] (2.25,1) rectangle (5.25,.1);
\node at (3.75,.55) {\small
	$\SSLVerify(\sk,\cdot,C)$};
\filldraw[fill=white] (2.25,-1) rectangle (4.25,-.1);
\node at (3.25,-.55) {\small
	$\{\Pi_x,I-\Pi_x\}$
};

\node at (5.75,.55) {$v$};
\node at (5.75,-.55) {$b$};

\draw[dashed] (.4,1.5) -- (2,1.5) -- (2,0) -- (4.35,0) -- (4.35,-1.1) -- (.4,-1.1) -- (.4,1.5);
\node at (.75,1.25) {$\adv$};

\node at (8.25,0) {
\begin{minipage}{3.75cm}
\centering
\underline{Winning Conditions:}\\
\vskip 5pt
$v=1$\\
$b=C(x)$
\end{minipage}
};
\end{tikzpicture}
\caption{The SSL game $\ExperimentSSL_{\adv,\SSL}$, where the behaviour of $\adv$ is specified by a CPTP map $\Phi_{\adv}$ and a set of two-outcome measurements $\{\Pi_x\}_{x\in\{0,1\}^n}$.}\label{fig:SSL-experiment}
\end{figure}

We base our security game, between a challenger (in this case a \emph{Lessor}) and an adversary $\adv$, on the $\ExperimentSSL$ from~\cite[Section 6]{CMP20}. The game is parametrized by a  distribution~$D$ over circuits in~$\mathcal{C}$, and a set of challenge distributions $\{D'_C\}_{C \in \calC}$ over inputs $\{0,1\}^n$.

\vspace{0.5em}

\begin{center}
\fbox{
\begin{minipage}{0.8\textwidth}
\textbf{The SSL game $\ExperimentSSL$}$_{\adv,\SSL}$
\begin{enumerate}
\item The Lessor samples $C \leftarrow D$ and runs $\SSLGen$ to obtain a secret key $\sk$. She then sends $\rho=\SSLLease(\sk,C)$ to $\adv$.

\item $\adv$ produces a state $\sigma$ on registers $\reg{YA}$ and sends register $\reg{Y}$ back to the Lessor and keeps $\reg{A}$.

\item (\emph{Verification phase.}) The Lessor runs $\SSLVerify$ on $\reg{Y}$, the circuit $C$ and the secret key $\sk$ and outputs the resulting bit $v$. If $\SSLVerify$ accepts ($v=1$), the game continues, otherwise it aborts and $\adv$ loses.

\item The Lessor samples an input $x \leftarrow D'_C$ and sends~$x$ to $\adv$.
\item $\adv$ returns a bit~$b$ to the Lessor.
\item The Lessor outputs 1 if and only if $b=C(x)$ and $v=1$, in which case, we say $\adv$ ``wins'' the game.
\end{enumerate}
\end{minipage}
}
\end{center}

\vspace{0.5em}

An adversary $\adv$ for $\ExperimentSSL$ can be described by: an arbitrary CPTP map $\Phi_{\adv}:{\cal L}(\reg{Y})\rightarrow{\cal L}(\reg{YA})$ for some arbitrary space $\reg{A}$, representing the action of $\adv$ in Step 2; and a set of two-outcome measurements $\{\Pi_x\}_{x\in\{0,1\}^n}$ on $\reg{A}$ such that given challenge $x$ in Step 4, $\adv$ obtains the bit $b$ in Step 5 by measuring $\reg{A}$ with $\{\Pi_x,I-\Pi_x\}$ (see \cref{fig:SSL-experiment}).

As in \Cref{sec:definition-CP}, we define security with respect to the trivial strategy where $\adv$ returns the program $\rho$ to the Lessor in Step 2, and tries to guess the most likely value for $b$, given input $x$.

Formally, we define $\pind_{D, \{D_C\}_{C \in \calC}}$ as follows.
The distributions $D$ and $\{D_C\}_{C \in \mc{C}}$ yield a joint
distribution $\tilde{D}$ on $\mc{C} \times \{0,1\}^n$ by first sampling
$C \gets D$ and then sampling $x \gets D_C$. Let $\hat{D}$ be the
marginal distribution of $x$ from $\tilde{D}$ and, for every $x'$, let
$\hat{D}_{x'}$ be the marginal
distribution of $C$ from $\tilde{D}$, conditioned on $x = x'$. Then,
\begin{equation}
\label{eq:pind}
	\pind_{D, \{D'_C\}_{C \in \calC}}
	=
\E_{x \gets \hat{D}} \max_{b \in \{0,1\}} \Pr_{C \gets \hat{D}_x}[C(x) = b].
\end{equation}

The above equation is very similar to $\pmarg$ given in
\Cref{eq:pmarg}. However, we point out that they are defined and used in
different contexts. Specifically, in $\ExperimentFree$ there are two
parties, Bob and Charlie, who must be challenged with inputs on which to
evaluate the function. However, there is only a single party attempting
to evaluate the function at the end of $\ExperimentSSL$. Thus, $\pmarg$
is defined with respect to the marginal distribution on Charlie's
challenge generated by the joint challenge distribution. On the other
hand, $\pind$ can be directly defined with respect to the single
challenge issued in $\ExperimentSSL$.

We now define the security of SSL as follows.

\begin{definition}[Security of SSL]
\label{defn:ssl-statistical-security}
An SSL scheme $\SSL$ for a set of $n$-bit circuits $\calC$ is \emph{$\epsilon$-secure with respect to the distribution $D$ and challenge distributions $\{D'_C\}_{C \in \calC}$} if for all adversaries $\adv$,
\begin{equation}
 \Pr[\ExperimentSSL_{\adv}] \leq \pind + \epsilon\,,
\end{equation}
where  $\pind = \pind_{D, \{D'_C\}_{C \in \mathcal{C}}}$.
\end{definition}

Observe that, as in the case with \Cref{defn:copy-protection}, our definition provides statistical guarantees for security as we impose no conditions on the adversaries.

%~~~~~~~~~~~~~~~~~~~~~~~~~~~~~~~~~~~~~~~~~~~~~~%
\subsection{Distributions for Point Functions} %
\label{sec:defn-distributions}                 %
%~~~~~~~~~~~~~~~~~~~~~~~~~~~~~~~~~~~~~~~~~~~~~~%

The definitions of correctness and security for copy protection and
secure software leasing presented earlier in this section are parametrized
by various distributions on the circuits that are encoded and the challenges
that are issued.

In this section, we define notation for the distributions we will consider in the setting of point functions. First, we will consider security in the setting when the point function is chosen uniformly at random.

\begin{definition}
\label{df:Dn}
We let $R$ be the uniform distribution on the set of
point functions $\{P_p \;:\; p \in \{0,1\}^n\}$. For simplicity, we will also use $R$ to simply refer to the uniform distribution on $\{0,1\}^n$, as we often conflate a point $p$ with its corresponding point function $P_p$.
\end{definition}

For a fixed point $p$, we will consider the distribution of inputs where $p$ is sampled with probability $1/2$, and otherwise, a uniform $x\neq p$ is sampled.
\begin{definition}
\label{df:Dy}
For any bit string $p \in \{0,1\}^n$, we define $\Dhalf_p$ to be the
distribution on $\{0,1\}^n$ such that
\begin{itemize}
	\item
		$p$ is sampled with probability $\frac{1}{2}$ and
	\item
		any $x \neq p$ is sampled with probability
		$\frac{1}{2}\cdot\frac{1}{2^n-1}$.
\end{itemize}
\end{definition}
This is a natural distribution in the setting of point functions, since it means that the function evaluates to a uniform random bit. This ensures that the output is non-trivial to guess --- an adversary's advantage against challenge
distributions of this form can be quantified by comparing it with their
probability of correctly guessing a random bit.
Furthermore, $\eta$-correctness with respect to this distribution, for some small $\eta$,  ensures that evaluating the point is correct except with small probability, and that all but a small fraction of the other inputs are evaluated correctly except with small probability.

%~~~~~~~~~~~~~~~~~~~~~~~~~~~~%
\section{Relationships Between Definitions} %
\label{sec:generic}
%~~~~~~~~~~~~~~~~~~~~~~~~~~~~%

In this section, we give some generic relationships between the definitions given in \cref{sec:definitions}. Specifically, in \cref{sec:mm-correct}, we show that any copy protection scheme for point functions that is secure in the malicious-malicious setting but only satisfies correctness with respect to the distribution family $\Dhalffam$, in which $\Dhalf_p$ samples $p$ with probability $1/2$ and all other strings uniformly, can be combined with a pairwise independent permutation family to get a scheme that is still secure in the malicious-malicious setting but is also correct with respect to any distribution (\cref{thm:mm-correct}). We recall that the malicious-malicious security setting is the standard security definition considered in previous works, and correctness with respect to any distribution is the standard notion of correctness. Thus, our construction given in \cref{sec:auth-CP}, while it has its advantages, falls short of achieving the standard security and correctness notions by being secure only in the honest-malicious setting, and by being correct only with respect to $\Dhalffam$. The results of \cref{sec:mm-correct} show that solving the former problem would also solve the latter.

Finally, in \cref{sec:SSL}, we describe how an honest-malicious copy protection scheme for any set of circuits $\calC$ can be turned into an SSL scheme for $\calC$ (\cref{thm:CP-to-SSL}). In particular, this means that the copy protection scheme for point functions presented in \cref{sec:auth-CP} implies an SSL scheme for point functions. We also describe how the latter can be extended into an SSL scheme for compute-and-compare programs (\cref{thm:SSLPF-to-SSLCC}).

\subsection{Malicious-Malicious Security and Correctness}\label{sec:mm-correct}

Let $\CP=(\Protect,\Eval)$ be a copy protection scheme
for point functions of length $n$ and fix a pairwise independent family
of permutations $\{h_r\}_{r\in{\cal R}}$ on the set $\{0,1\}^n$.
We define another copy protection scheme for point functions of
length $n$, denoted
${\sf MIX}^{\CP}=({\sf MIX^{\CP}.Protect},{\sf MIX^{\CP}.Eval})$, as
follows:
\begin{description}
\item[\rm${\sf MIX^{CP}.Protect}(p)$:] On input of $p$ (representing the point function $P_p$), output
\begin{align*}
\sum_{r\in{\cal R}}\frac{1}{|{\cal R}|} \ketbra{r}\otimes {\Protect}(h_r(p)).
\end{align*}
\item[\rm${\sf MIX^{CP}.Eval}((r,\sigma),x)$:] On input $x$ and program $(r,\sigma)$, output ${\Eval}(\sigma,h_r(x))$.
\end{description}
We call a set of challenge distributions $\{D_p\}_{p\in\{0,1\}^n}$ \emph{symmetric} if for $p\leftarrow R$, where we recall that $R$ is the uniform distribution on points, and $(x_1,x_2)\leftarrow  D_p$, the probability of any triple $(p,x_1,x_2)$ is the same as $(\pi(p),\pi(x_1),\pi(x_2))$ for any permutation $\pi$ on $\{0,1\}^n$. Equivalently, $D_p(x_1,x_2)$ can only depend on whether $x_1=x_2$, whether $x_1=p$ and whether $x_2=p$. In particular, the set of product distributions $\{\Dhalf_p\times \Dhalf_p\}_p$ is symmetric.

In the remainder of this section, we show the following:
\begin{theorem}\label{thm:mm-correct}
If the scheme $\CP$ is $\epsilon$-malicious-malicious secure with
respect to the uniform distribution on points $R$ and any symmetric set of challenge distributions $\{D_p\}_{p\in\{0,1\}^n}$,
and $\eta$-correct with respect to the distribution family $\Dhalffam$,
then $\sf MIX^{CP}$ is $\epsilon$-malicious-malicious secure with respect
to $R$ and $\{D_p\}_{p\in \{0,1\}^n}$ and $2\eta$-correct with respect to any distribution.
\end{theorem}

We begin by showing that $\sf MIX^{CP}$ is $2\eta$-correct:

\begin{lemma}
If the scheme $\CP$ is $\eta$-correct with respect to the distribution family $\Dhalffam$, then ${\sf MIX^{CP}}$ is $2\eta$-correct with respect to any distribution family.
\end{lemma}
\begin{proof}
For any $p,x\in\{0,1\}^n$, and bit $b$, the probability that
$\Eval$ outputs $b$ on input $x$, when evaluating the program
${\Protect}(p)$ is
$\Tr(\ketbra{b}{b} {\Eval}({\Protect}(p),x))$.
By the $\eta$-correctness of $\CP$ under the distribution $\Dhalffam$,
we have, for any point~$y\in\{0,1\}^n$:
\begin{align*}
\frac{1}{2}\Tr\left(\ketbra{1}{1} {\Eval}({\Protect}(y),y) \right)+\frac{1}{2}\sum_{x\neq y}\frac{1}{2^n-1}\Tr\left(\ketbra{0}{0} {\Eval}({\Protect}(y),x) \right)
&\geq 1-\eta,
\end{align*}
so
\begin{align}
\frac{1}{2}\Tr\left(\ketbra{1}{1} {\Eval}({\Protect}(y),y) \right)
&\geq \frac{1}{2}-\eta,\label{eq:delta-corr1}\\
\mbox{and }\;\;\frac{1}{2}\sum_{x\neq y}\frac{1}{2^n-1}\Tr\left(\ketbra{0}{0} {\Eval}({\Protect}(y),x) \right)
&\geq \frac{1}{2}-\eta.\label{eq:delta-corr2}
\end{align}

Fix $p$. The probability that ${\sf MIX^{CP}.Eval}({\sf MIX^{CP}.Protect}(p),p)$ outputs the correct value of $1$ is:
\begin{align*}
& \Tr\left( \ketbra{1}{1}{\sf MIX^{CP}.Eval}({\sf MIX^{CP}.Protect}(p),p) \right)\\
={}& \sum_{r\in{\cal R}}\frac{1}{|{\cal R}|}\Tr\left(\ketbra{1}{1}{\sf MIX^{CP}.Eval}((r,{\Protect}(h_r(p))),p)\right)\\
={}& \sum_{r\in{\cal R}}\frac{1}{|{\cal R}|}\Tr\left(\ketbra{1}{1}{\Eval}({\Protect}(h_r(p))),h_r(p))\right)\\
\geq & 1-2\eta, & \mbox{(by \eqref{eq:delta-corr1}).}
\end{align*}

For any $p'\neq p$, the probability that ${\sf MIX^{CP}.Eval}({\sf MIX^{CP}.Protect}(p),p')$ outputs the correct value of $0$~is:
\begin{align*}
& \Tr\left(\ketbra{0}{0}{\sf MIX^{CP}.Eval}({\sf MIX^{CP}.Protect}(p),p') \right)\\
={}& \sum_{r\in{\cal R}}\frac{1}{|{\cal R}|}\Tr\left(\ketbra{0}{0}{\Eval}({\Protect}(h_r(p)),h_r(p'))\right)\\
={}& \frac{1}{|{\cal R}|}\sum_{y\in\{0,1\}^n}\sum_{x\neq y}\sum_{\substack{r:h_r(p)=y,\\ h_r(p')=x}}\Tr\left(\ketbra{0}{0}{\Eval}({\Protect}(y),x)\right)\\
={}&\frac{1}{2^n}\sum_{y\in\{0,1\}^n}\frac{1}{2^n-1}\sum_{x\neq y}\Tr\left(\ketbra{0}{0}{\Eval}({\Protect}(y),x)\right) & \mbox{(by pairwise independence)}\\
\geq & 1-2\eta,
\end{align*}
by \eqref{eq:delta-corr2}, completing the proof.
\end{proof}

For the security proof, we will actually show that any mixture of malicious-malicious secure schemes is malicious-malicious secure (\cref{thm:MX}). We first show that if ${\Protect}(p)$ is an $\epsilon$-malicious-malicious secure encoding, then for each $r$, the scheme that encodes $p$ as ${\Protect}(h_r(p))$ is also $\epsilon$-malicious-malicious secure (\cref{lem:hrp}). The combination of these two facts completes the proof of \cref{thm:mm-correct}, since $\sf MIX^{CP}$ is a mixture, in the sense of \cref{thm:MX}, of schemes of the form described in \cref{lem:hrp}.

\begin{lemma}\label{lem:hrp}
Suppose the scheme $\CP$ is $\epsilon$-malicious-malicious secure with respect to the uniform distribution on points $R$ and any symmetric set of challenge distributions $\{D_p\}_{p\in\{0,1\}^n}$, and let $\pi$ be any permutation on $\{0,1\}^n$. Then if ${\sf CP'.Protect}(p)={\Protect}(\pi(p))$, ${\CP}'$ is $\epsilon$-malicious-malicious secure with respect to $D$ and~$\{D_p\}_p$.
\end{lemma}
\begin{proof}
Let $(\pirate,\adv_1,\adv_2)$ be an adversary for ${\CP}'$ in $\ExperimentFree$ (see \cref{fig:freeloading-experiment}) with success probability $q$, where $\pirate$'s action is given by the CPTP map $\Phi_{\cal P}:{\cal L}({\sf Y})\rightarrow {\cal L}({\sf A}_1{\sf A}_2)$; the action of $\adv_1$ (Bob) is described by a set of two-outcome measurements $\{\Pi_x\}_{x\in\{0,1\}^n}$ on $\reg{A}_1$, such that on challenge input $x_1$, Bob does the measurement $\{\Pi_{x_1},I-\Pi_{x_1}\}$ on $\reg{A}_1$ to obtain the bit $b_1$; and similarly, the action of $\adv_2$ (Charlie) is described by a set of two-outcome measurements $\{\Lambda_x\}_{x\in \{0,1\}^n}$ on $\reg{A}_2$.
Let $\Pi_{x_1}^1=\Pi_{x_1}$ and $\Pi_{x_1}^0=I-\Pi_{x_1}$, so that $\Pi_{x_1}^{P_p(x_1)}$ projects onto the part of Bob's input that leads Bob to output the correct answer;
and similarly define $\Lambda_{x_2}^b$. Then, since $p$ is sampled with probability $R(p)=1/2^n$, and given that $p$ is sampled, a challenge $(x_1,x_2)$ is sampled with probability $D_p(x_1,x_2)$, we have:
\begin{align*}
q={}& \sum_{p,x_1,x_2}\frac{1}{2^n}D_p(x_1,x_2)\Tr\left(((\Pi_{x_1}^{P_p(x_1)})^{{\sf A}_1}\otimes(\Lambda_{x_2}^{P_p(x_2)})^{{\sf A}_2})\Phi_{\cal P}({\CP}'{\sf .Protect}(p))\right) \\
={}& \sum_{p,x_1,x_2}\frac{1}{2^n}D_p(x_1,x_2)\Tr\left((\Pi_{x_1}^{P_p(x_1)}\otimes\Lambda_{x_2}^{P_p(x_2)})\Phi_{\cal P}({\Protect}(\pi(p)))\right) \\
={}& \sum_{p,x_1,x_2}\frac{1}{2^n}D_p(\pi^{-1}(x_1),\pi^{-1}(x_2))\Tr\left((\Pi_{\pi^{-1}(x_1)}^{P_{\pi^{-1}(p)}(\pi^{-1}(x_1))}\otimes\Lambda_{\pi^{-1}(x_2)}^{P_{\pi^{-1}(p)}(\pi^{-1}(x_2))})\Phi_{\cal P}({\Protect}(p))\right) \\
={}& \sum_{p,x_1,x_2}\frac{1}{2^n}D_p(x_1,x_2)\Tr\left((\Pi_{\pi^{-1}(x_1)}^{P_p(x_1)}\otimes\Lambda_{\pi^{-1}(x_2)}^{P_p(x_2)})\Phi_{\cal P}({\Protect}(p))\right),
\end{align*}
where in the last step, we used the symmetry of $\{D_p\}_p$.
Then letting $\Pi_x'=\Pi_{\pi^{-1}(x)}$ and $\Lambda_x'=\Lambda_{\pi^{-1}(x)}$, and defining $\adv_1'$ and $\adv_2'$ so that their respective strategies are to perform the measurements $\{\Pi_{x_1}',I-\Pi_{x_1}'\}$ and $\{\Lambda_{x_2}',I-\Lambda_{x_2}'\}$ upon receiving their respective challenges $x_1$ and $x_2$, we have that $(\pirate,\adv_1',\adv_2')$ is an adversary for $\CP$ with success probability $q$.
Thus, we must have $q\leq \frac{1}{2}+\epsilon$.
\end{proof}

\begin{theorem}\label{thm:MX}
Let $D$ be any distribution on points, and $\{D_p\}_p$ any set of challenge distributions.
Suppose $\{{\CP}_r=({\CP}_r.{\sf Protect},{\CP}_r.{\sf Eval})\}_{r\in{\cal R}}$ is a family of point function encoding schemes, each with $\epsilon$-malicious-malicious security with respect to  $D$ and $\{D_p\}_p$. Define a scheme ${\sf MX}$ by
\begin{description}
\item[\rm${\sf MX.Protect}(p)$:] On input $p$, output $\sum_{r\in{\cal R}}\frac{1}{|{\cal R}|} \ketbra{r}\otimes {\CP}_r.{\sf Protect}(p)$.
\item[\rm${\sf MX.Eval}((r,\sigma),x)$:] On input $x$ and program $(r,\sigma)$, output ${\CP}_r.{\sf Eval}(\sigma,x)$.
\end{description}
Then ${\sf MX}$ is $\epsilon$-malicious-malicious secure with respect to $D$ and $\{D_p\}_p$.
\end{theorem}
\begin{proof}
Let $(\pirate,\adv_1,\adv_2)$ be an adversary for ${\sf MX}$ in $\ExperimentFree$ that succeeds with probability $q$, where $\pirate$'s action is given by the CPTP map
$\Phi_{\cal P}:{\cal L}({\sf Y})\rightarrow {\cal L}({\sf A}_1{\sf A}_2)$;
the action of $\adv_1$ (Bob) is described by a set of two-outcome measurements $\{\Pi_x\}_{x\in\{0,1\}^n}$ on $\reg{A}_1$, such that on challenge input $x_1$, Bob does the measurement $\{\Pi_{x_1},I-\Pi_{x_1}\}$ on $\reg{A}_1$ to obtain the bit $b_1$; and similarly, the action of $\adv_2$ (Charlie) is described by a set of two-outcome measurements $\{\Lambda_x\}_{x\in \{0,1\}^n}$ on $\reg{A}_2$.
Let $\Pi_{x_1}^1=\Pi_{x_1}$ and $\Pi_{x_1}^0=I-\Pi_{x_1}$, so that $\Pi_{x_1}^{P_p(x_1)}$ projects onto the part of Bob's input that leads Bob to output the correct answer;
and similarly define $\Lambda_{x_2}^b$. Then, since $p$ is sampled with probabilty $D(p)$, and given that $p$ is sampled, a challenge $(x_1,x_2)$ is sampled with probability $D_p(x_1,x_2)$, we have:
\begin{align*}
q={}& \sum_{p,x_1,x_2}D(p)D_p(x_1,x_2)\Tr\left(((\Pi_{x_1}^{P_p(x_1)})^{{\sf A}_1}\otimes(\Lambda_{x_2}^{P_p(x_2)})^{{\sf A}_2})\Phi_{\cal P}({\sf MX.Protect}(p))\right) \\
={}& \sum_{r\in {\cal R}}\frac{1}{|{\cal R}|}\sum_{p,x_1,x_2}D(p)D_p(x_1,x_2)\Tr\left((\Pi_{x_1}^{P_p(x_1)}\otimes\Lambda_{x_2}^{P_p(x_2)})\Phi_{\cal P}(\ketbra{r}\otimes {\CP}_r.{\sf Protect}(p))\right),
\end{align*}
so there exists $r\in {\cal R}$ such that
\begin{align*}
q\leq {}& \sum_{p,x_1,x_2}D(p)D_p(x_1,x_2)\Tr\left((\Pi_{x_1}^{P_p(x_1)}\otimes\Lambda_{x_2}^{P_p(x_2)})\Phi_{\cal P}(\ketbra{r}{r}\otimes {\CP}_r.{\sf Protect}(p))\right)\\
={}&\sum_{p,x_1,x_2}D(p)D_p(x_1,x_2)\Tr\left((\Pi_{x_1}^{P_p(x_1)}\otimes\Lambda_{x_2}^{P_p(x_2)})\Phi_r({\CP}_r.{\sf Protect}(p))\right),
\end{align*}
where $\Phi_r(\cdot) := \Phi_{\cal P}(\ketbra{r}{r}\otimes \cdot)$. Thus, if $\pirate'$ is an adversary who performs the map $\Phi_r$, $(\pirate',\adv_1,\adv_2)$ is an adversary for ${\CP}_r$ that succeeds with probability at least $q$, and so we must have $q\leq \frac{1}{2}+\epsilon$.
\end{proof}

\subsection{Secure Software Leasing from Honest-Malicious Copy Protection}
\label{sec:SSL}

In this section, we show that honest-malicious copy protection for point
functions implies secure software leasing for compute-and-compare
programs.

In fact, following~\Cref{fig:links}, we show in \Cref{sec:Copy-protect-implies-SSL} how an honest-malicious copy protection scheme for some set of functions $\calC$ can be used to create an SSL scheme for $\calC$ (\Cref{thm:CP-to-SSL}).  Next, in~\Cref{sec:SSLPF-to-SSLCC}, we use a result from~\cite{CMP20} along with our definitions of correctness and security for SSL schemes to show an SSL scheme for point functions implies an SSL scheme for compute-and-compare programs (\Cref{thm:SSLPF-to-SSLCC}). Putting these two together, we can demonstrate that honest-malicious copy protection for point functions implies secure software leasing for compute-and-compare functions.

\subsubsection{From Honest-Malicious Copy Protection to SSL}
\label{sec:Copy-protect-implies-SSL}

Here, we show that a copy protection scheme for a set of Boolean
circuits $\mc{C}$ on $n$-bits that is correct with respect to \emph{two} families of distributions, $\{T_C\}_{C\in\calC}$ and $\{T_C'\}_{C\in\calC}$, and honest-malicious secure with
respect to the circuit distribution $D$ on $\calC$, and the challenge distributions $\{T_C' \times T''_C\}_{C \in \mc{C}}$,
can be used to construct an SSL scheme for $\mc{C}$ that is correct with respect to $\{T_C\}_{C\in\mc{C}}$ and secure with
respect to $D$ and $\{T''_C\}_{C \in \mc{C}}$. Here, $T_C'\times T_C''$ denotes the product distribution of the two distributions $T_C'$ and $T_C''$, which are both distributions on $\{0,1\}^n$.

Let $\CPC = (\CPCProt, \CPCEval)$ be a copy protection scheme for a set
of $n$-bit Boolean circuits~$\calC$. We define an SSL scheme for $\calC$, $\SSLC = (\SSLCGen, \SSLCLease, \SSLCEval, \SSLCVerify)$ as follows:
\begin{description}
	\item[\rm$\SSLCGen$:] Output an empty secret key $\sk = \emptyset$.
	\item[\rm$\SSLCLease(C)$:] As the secret key is empty, the only input is the circuit $C$. On input $C$, output $\rho = \CPCProt(C)$.
	\item[\rm$\SSLCEval(\rho, x)$:] On input $\rho\in{\cal D}(\reg{Y})$ and $x \in \bool^{\len}$, run the program-preserving version of $\CPCEval$, $\overline{\CPCEval}(\rho,x)$, described in \cref{sec:reuse} and output the resulting bit $b$ and post-evaluated program $\tilde{\rho}\in\reg{Y}$.
	\item[\rm$\SSLCVerify(C, \sigma)$:] As the secret key is empty, the only inputs are the circuit $C$ and a state $\sigma\in\reg{Y}$. Sample $x \leftarrow T_C'$ and output 1 if and only if $\CPCEval(\sigma, x)$ is $C(x)$.
\end{description}

We prove the following.
\begin{theorem}
\label{thm:CP-to-SSL}
Suppose the scheme $\CPC$ is a copy protection scheme for circuits $\calC$, that is $\eta$-correct with respect to
$\{T_C\}_{C \in \mc{C}}$, $\eta$-correct with respect
to $\{T'_C\}_{C \in \mc{C}}$, and $\epsilon$-honest-malicious secure with
respect to the distribution $D$ on $\calC$ and challenge distributions $\{T'_C \times T''_C\}_{C \in \mc{C}}$. Then the scheme
$\SSLC$, constructed from $\CPC$ as described above, is an SSL scheme
for $\calC$ that is $\eta$-correct with respect to
$\{T_C\}_{C \in \mc{C}}$ and $\epsilon$-secure with respect to the distributions
$D$ and $\{T''_C\}_{C \in \mc{C}}$.
\end{theorem}

We prove correctness and security separately. We begin with correctness.

\begin{lemma}
If the scheme $\CPC$ is $\eta$-correct with respect to a family of distributions $\{T_C\}_{C\in\calC}$ and also with respect to the family $\{T_C'\}_{C\in\calC}$ used in the definition of $\SSLC$, then the scheme $\SSLC$ described above is $\eta$-correct with respect to $\{T_C\}_{C\in\calC}$.
\end{lemma}

\begin{proof}
	The proof proceeds directly by considering the definition of correctness for $\CPC$ and $\SSLC$ from~\Cref{def:cp-correct,def:ssl-correct}. For any $C \in \calC$, if $\rho = \SSLCLease(C) = \CPCProt(C)$, then
	\begin{align}
	\label{eq:corr1}
		\E_{x \leftarrow T_C} \Tr\left(\ketbra{C(x)}\SSLCEval(\rho, x)\right) = \E_{x \leftarrow T_C} \Tr\left(\ketbra{C(x)}\CPCEval(\rho, x)\right) \geq 1 - \eta,
	\end{align}
	where the last inequality uses that $\CPC$ is $\eta$-correct with respect to $\{T_C\}_{C\in\calC}$.
	
	To satisfy the correctness requirement for $\SSLCVerify$, note that by the construction of $\SSLCVerify$,
	 \begin{align}
	 \label{eq:corr2}
		\Tr\left(\ketbra{1}\SSLCVerify(C, \rho)\right) = \E_{x \leftarrow T_C'} \Tr\left(\ketbra{C(x)}\CPCEval(\rho, x)\right) \geq 1 - \eta,
	\end{align}
	where the last inequality uses that $\CPC$ is also $\eta$-correct with respect to $\{T_C'\}_{C\in\calC}$. Putting together~\Cref{eq:corr1,eq:corr2}, we can conclude that $\SSLC$ is $\eta$-correct with respect to $\{T_C\}_{C\in\calC}$.
\end{proof}

We move on to the security guarantees for $\SSLC$. Observe that, in an SSL scheme, the challenge is sent only after the leased copy is returned. This is in contrast to a copy protection scheme where the challenger does not see the adversary's output registers while sampling the challenge questions. However, the Lessor gains no advantage from this, since by definition of the security game $\ExperimentSSL$, the Lessor samples the challenge $x$ according to some fixed distribution, independent of what she receives from the adversary.

\begin{lemma}
\label{thm:CP_PF-to-SSL_PF}
If the scheme $\CPC$ is $\epsilon$-honest-malicious secure with respect to a circuit distribution~$D$, and challenge distributions $\{T_C'\times T_C''\}_{\calC}$, then the scheme $\SSLC$ described above is $\epsilon$-secure with respect to $D$ and $\{T_C''\}_{C\in\calC}$.
\end{lemma}

The main intuition for the proof is to map the honest evaluation in the scheme $\CPC$ to the Lessor's verification procedure in the scheme $\SSLC$. The $\epsilon$-correctness of $\CPCEval$ ensures that the verification is accepted with sufficiently high probability. Next, we map the malicious user Charlie's ($\adv_2$) evaluation in $\ExperimentFree$ to the adversary's evaluation in $\ExperimentSSL$. Assuming that $\CPC$ is secure, we can bound Charlie's probability of guessing the right answer, which in turn bounds the adversary's probability of guessing the right answer. Putting it together, we can conclude that the corresponding SSL scheme $\SSLC$ is secure.

\begin{proof}
Let $\pind = \pind_{D, \{T_C''\}_{C\in\calC}}$ and $\pmarg = \pmarg_{D, \{T_C'\times T_C''\}_{C\in\calC}}$ and assume for the rest of the proof that $\ExperimentSSL$ is instantiated with circuit distribution $D$ and challenge ensemble $\{T_C''\}_{C\in\calC}$, and $\ExperimentFree$ is instantiated with circuit distribution $D$ and challenge ensemble $\{T_C'\times T_C''\}$.

The proof proceeds by contradiction i.e., we use a winning adversary against the scheme $\SSL$, $\adv_{\SSL}$ to construct a winning honest-malicious adversary for the scheme $\CPC$, $\hatAdv_{\CP}$.

Let $\adv_{\SSL}$ be an adversary for $\ExperimentSSL$, and suppose
\begin{equation}
	\Pr\big[ \ExperimentSSL_{\adv_{\SSL},\SSL} \big] > \pind + \epsilon \,.
\end{equation}
We show how to construct $\hatAdv_{\CP}$ that wins $\ExperimentFree_{\hatAdv_{\CP},\CP}$ with probability $> \pmarg + \epsilon$.

The behaviour of the adversary $\adv_{\SSL}$ can be described in two parts  (see also \cref{fig:SSL-experiment}).
First, the adversary applies an arbitrary CPTP map $\Phi_{\adv_{\SSL}}:{\cal L}(\reg{Y})\rightarrow{\cal L}(\reg{YA})$ to his input $\rho=\SSLCLease(\sk,C)$, sending the $\reg{Y}$ part to the Lessor
(Step 2 of $\ExperimentSSL$); and keeping the $\reg{A}$ part, for some arbitrary space $\reg{A}$, for himself. Later, when $\adv_{\SSL}$ receives the challenge $x$, he uses it to select a two outcome measurement $\{\Pi_x,I-\Pi_x\}$ with which to measure his register $\reg{A}$ to obtain a bit $b$ (Step 5 of $\ExperimentSSL$).
Construct an honest-malicious adversary $\hatAdv_{\CP} = (\pirate, \CPCEval, \adv_2)$ such that:
\begin{itemize}
\item $\pirate$'s behaviour is described by the map $\Phi_{\pirate}=\Phi_{\adv_{\SSL}}$, and
\item $\adv_2$'s behaviour is described by the two-outcome measurements $\{\Pi_x\}_{x\in\{0,1\}^n}$.
\end{itemize}
Then $\ExperimentFree_{\hatAdv_{\CP},\CP}$ proceeds as follows.
\begin{itemize}
\item The challenger samples $C \leftarrow D$ and sends $\rho = \CPCProt(C)$ to $\pirate$.
\item Upon receiving $\rho$, $\pirate$ computes $\sigma=\Phi_{\adv_{\SSL}}(\rho)\in {\cal D}(\reg{YA})$ and sends $\reg{Y}$ to Bob, who is controlled by the challenger in the honest-malicious setting, and $\reg{A}$ to Charlie.
\item The challenger samples $x_1 \leftarrow T'_C$ and sends it to Bob, and samples $x_2 \leftarrow T''_C$ and sends it to Charlie.
\item Bob runs $\CPCEval$ on ${\reg{Y}}$ and $x_1$ and outputs the resulting bit $b_1$.
\item Charlie measures $\reg{A}$ using the measurement $\{\Pi_{x_2},I-\Pi_{x_2}\}$ and outputs the resulting bit $b_2$.
\item The challenger outputs 1 if and only if $b_1 = C(x_1)$ and $b_2 = C(x_2)$.
\end{itemize}

To see why this construction works, observe that the pirate $\pirate$ is essentially acting as $\adv_{\SSL}$ in Step 2 of $\ExperimentSSL$, the only difference is that after applying $\Phi_{\adv_{\SSL}}$, $\adv_{\SSL}$ keeps the register $\reg{A}$ for himself, whereas $\pirate$ sends it to Charlie. Later Charlie behaves just as $\adv_{\SSL}$ behaves in Step 5 of $\ExperimentSSL$.
If we define $\Pi_x^1=\Pi_x$ and $\Pi_x^0=I-\Pi_x$, then $\Pi_{x_2}^{C(x_2)}$ is the projector onto the part of Charlie's input state that will lead Charlie to output the correct bit, $b_2=C(x_2)$ in $\ExperimentFree$. It's also the projector onto the part of $\adv_{\SSL}$'s memory $\reg{A}$
that leads to him outputting the correct bit $b=C(x)$ in $\ExperimentSSL$. Then, letting $\Psi_{\sf Ver}^C$ denote the map induced by $\SSLCVerify(C,\cdot)$, we have:
\begin{align*}
\Pr[\ExperimentSSL_{\adv_{\SSL},\SSLC}] &=\E_{C\leftarrow D}\E_{\substack{x_2\leftarrow T_C''}} \Tr\left((\ketbra{1}\otimes \Pi_{x_2}^{C(x_2)})(\Psi_{\sf Ver}^{C}\otimes \Id)\circ\Phi_{\adv_{\SSL}}(\SSLCLease(C)) \right)\\
&=\E_{C\leftarrow D}\E_{\substack{x_2\leftarrow T_C''}} \Tr\left((\ketbra{1}\otimes \Pi_{x_2}^{C(x_2)})(\Psi_{\sf Ver}^{C}\otimes \Id)\circ\Phi_{\pirate}(\CPCProt(C)) \right).
\end{align*}
Let $\Psi_{\sf Eval}^x$ denote the map induced by $\CPCEval(\cdot,x)$. Then $\Psi_{\sf Ver}^C$ works by sampling $x_1\leftarrow T_C'$ and applying $\Psi_{\sf Eval}^{x_1}$, and outputting 1 if and only if the result is $C(x)$. Thus, we can continue from above:
\begin{align*}
\Pr[\ExperimentSSL_{\adv_{\SSL},\SSLC}] &=\E_{C\leftarrow D}\E_{\substack{x_1\leftarrow T_C'\\ x_2\leftarrow T_C''}} \Tr\left((\ketbra{C(x_1)}\otimes \Pi_{x_2}^{C(x_2)})(\Psi_{\sf Eval}^{x_1}\otimes \Id)\circ\Phi_{\pirate}(\CPCProt(C)) \right)\\
&=\Pr[\ExperimentFree_{\hatAdv_{\CP},\CPC}].
\end{align*}
By assumption, $\Pr[\ExperimentSSL_{\adv_{\SSL},\SSLC}] > \pind + \epsilon$, so
$\Pr[\ExperimentFree_{\hatAdv_{\CP},\CPC}] > \pind+\epsilon$.

Additionally, for any $C$, $x_2$ is distributed according to $T_C''$, independent of $x_1$. This is also the challenge distribution in $\ExperimentSSL$.
Therefore, using~\Cref{eq:pind,eq:pmarg}, $\pind = \pmarg$ which implies that,
\begin{equation*}
	\Pr[\ExperimentFree_{\hatAdv_{\CP},\CPC}] > \pmarg + \epsilon.
\end{equation*}
This contradicts the assumption that $\CP$ is $\epsilon$-honest-malicious secure and completes the proof.
\end{proof}

We remark that the previous proof did not make any assumptions about the abilities of the adversaries. Hence, if the copy protection scheme $\CPC$ achieves statistical security guarantees, then so does the corresponding SSL scheme $\SSLC$.

\subsubsection{From SSL for Point Functions to SSL for Compute-and-Compare Programs}
\label{sec:SSLPF-to-SSLCC}
In this section we present a restatement of a theorem due to \cite{CMP20}, which states that an SSL scheme for point functions that is $\epsilon$-secure with respect to a family of distributions can be modified to get an SSL scheme for compute-and-compare programs that is also $\epsilon$-secure with respect to a related family of distributions. In the spirit of the results from the previous section, we state this result with a more precise relationship between the distributions used for the point functions and the compute-and-compare programs.

Let $F$ denote any set of functions from $\{0,1\}^n$ to $\{0,1\}^m$. Let
$\calF = \{(f,y):f\in F,y\in\{0,1\}^m\}$ be the set of compute-and-compare circuits for $F$, where as with point functions, we conflate $(f,y)$ with a circuit $\CC_y^f$ for the function that outputs 1 on input $x$ if and only if $f(x)=y$.

Let $\SSLPF =$ ($\SSLPFGen$, $\SSLPFLease$, $\SSLPFEval$, $\SSLPFVerify$) be an SSL scheme for $m$-bit point functions. We define an SSL scheme for compute-and-compare functions $\calF$, $\SSLCC =$ ($\SSLCCGen$, $\SSLCCLease$, $\SSLCCEval$, $\SSLCCVerify$) as follows:
\begin{description}
\item[\rm$\SSLCCGen$:] Compute $\SSLPFGen$ and output the resulting secret key $\sk$.
\item[\rm$\SSLCCLease(\sk, (f, y))$:] On input secret key $\sk$ and $(f,y)\in {\cal F}$, output $(f, \rho)$ where $\rho = \SSLPFLease(\sk, P_y)$.
\item[\rm$\SSLCCEval ( (f, \rho), x )$:] On input $\rho\in{\cal D}(\reg{Y})$, $f\in F$, and $x \in \bool^n$ do the following:
	\begin{enumerate}
		\item Compute $y' = f(x)$.
		\item Compute $\SSLPFEval(\rho, y')$ to get an output bit $b$ and post-evaluated state $\tilde{\rho}$, and output $b$ and $(f,\tilde\rho)$.
	\end{enumerate}
\item[\rm$\SSLCCVerify (\sk, (f, y), \sigma)$:] On input secret key $\sk$, $(f,y)\in \cal F$ and $\sigma\in{\cal D}(\reg{Y})$, compute $\SSLPFVerify(\sk, P_y, \sigma)$ and output the resulting bit $v$.
\end{description}

Formally, we show the following theorem. The proof of correctness
follows directly from definitions and the security proof follows the
same lines as the one presented in \cite{CMP20}. They are given in \cref{sc:gen-proofs}.
\begin{theorem}
\label{thm:SSLPF-to-SSLCC}
We fix the following distributions.
\begin{itemize}
\item $D$: A distribution over compute-and-compare functions $\CC_y^f$, or equivalently, over $(f,y)\in\calF$. Fixing a function $f\in F$ induces a marginal distribution $D_f$ over $y\in\{0,1\}^m$, or equivalently, over $m$-bit point functions $P_y$.
\item $\{T_{f,y}^{CC}\}_{f,y}$ and $\{D_{f,y}^{CC}\}_{f,y}$: Families of distributions over inputs $x\in\{0,1\}^n$ to compute-and-compare functions $\CC_y^f$.
\item $\{T_{f,y}^{PF}\}_{f,y}$ and $\{D_{f,y}^{PF}\}_{f,y}$: Families of distributions over inputs $z\in\{0,1\}^m$ to $m$-bit point functions $P_y$, where $T_{f,y}^{PF}$ is defined from $T_{f,y}^{CC}$ by sampling $x\leftarrow T^{CC}_{f,y}$ and outputting $f(x)$; and $D_{f,y}^{PF}$ is defined similarly from $D_{f,y}^{CC}$.
\end{itemize}
Suppose that $\SSLPF$ is a secure software leasing scheme for point
functions such that, for every $f\in F$, $\SSLPF$ is $\eta$-correct with
respect to the distribution family $\{T_{f,y}^{PF}\}_{y\in\{0,1\}^m}$
and $\epsilon_f$-secure with respect to the circuit distribution $D_f$
and challenge distributions $\{D_{f, y}^{PF}\}_{y\in\{0,1\}^m}$ where
\begin{equation}
\label{eq:eps_f}
	\epsilon_f
	=
	\left(
		p^{\text{triv}}_{D,\{D^{CC}_{f,y}\}_{(f,y)}}
		-
		p^{\text{triv}}_{D_{f^*},\{D^{PF}_{f^*,y}\}_{y}}
	\right)
	+
	\epsilon.
\end{equation}
Then the scheme $\SSLCC$, constructed from $\SSLPF$ as described above, is an SSL scheme for compute-and-compare programs in $\calF$ that is $\eta$-correct with respect to the family $\{T_{f,y}^{CC}\}_{(f,y)\in \calF}$ and $\epsilon$-secure with respect to program distribution $D$ and challenge distributions $\{D_{f,y}^{CC}\}_{(f,y)\in\calF}$.
\end{theorem}

%=====================================================%
\section{Authentication-based Copy Protection Scheme} %
\label{sec:auth-CP}                                   %
%=====================================================%

In this section, we show how to construct a copy protection scheme for
point functions, with honest-malicious security, from a total authentication scheme.

Recall that we assume that our circuits are searchable, which, for point
functions, implies that there is an efficient algorithm which can
produce the point $p$ from a circuit which computes its point function.
Thus, we will freely identify circuits for the point function $P_p$
simply with $p$. Specifically, our copy protection scheme will take as
input a point $p$ instead of a circuit.

%~~~~~~~~~~~~~~~~~~~~~~~~~~~~~~~~~~~~~~~~~%
\subsection{Construction and Correctness} %
%~~~~~~~~~~~~~~~~~~~~~~~~~~~~~~~~~~~~~~~~~%

Let
$\tsf{QAS} = \left(\tsf{QAS.Auth}, \tsf{QAS.Ver}\right)$ be an
$\epsilon$-total quantum authentication scheme, as in \cref{df:QAS},
with $\epsilon \leq \frac{1}{2}$ for a message space $\tsf{M}$ of
dimension greater than or equal to two with key set $\mc{K}=\{0,1\}^n$.
Fix some state $\ket{\psi} \in \tsf{M}$.

We recall that we assume that for every key $k$, the action of
$\tsf{QAS.Auth}$ with this key can be modeled by an isometry
$A_k : \tsf{M} \to \tsf{Y}$. Note that since $A_k$ is an isometry,
$A_kA_k^\dagger$ is the projector onto $\text{im}(A_k)$. Further, let
$V_k : \reg{Y} \to \reg{MFX}$ be an isometry which purifies the CPTP map
$\tsf{QAS.Ver}_k$ defined in \cref{eq:QAS-assumption-ver}, where the register
$\reg{X}$ corresponds to the Hilbert space used for this purification.
To simplify our notation, we will absorb $\sf X$ into the flag register,
which we no longer assume to be two-dimensional. We can still assume
that there is a unique accepting state $\ket{\text{Acc}}\in \sf F$.%
\footnote{This follows from correctness, since for every state
$\ket{\psi}$, we necessarily have $
	V_kA_k\ket{\psi}
	=
	\ket{\text{Acc}}_{\sf F}
	\ket{\psi}_{\reg{M}}
	\ket{X_\psi}_{\reg{X}}
$ for some state $\ket{X_\psi}$, and by the fact that $V_kA_k$ must
preserve inner products, we necessarily have $\ket{X_\psi}=\ket{X}$
independent of~$\ket{\psi}$. Thus, we can let
$\ket{\text{Acc}}_{\reg{F}}\ket{X}_{\reg{X}}$ be the accepting state on
$\reg{FX}$.} Thus, from here on, we assume that
$V_k:\reg{Y}\to \reg{MF}$ is an isometry, and $\reg{F}$ has dimension at
least two (but possibly larger) with $\ket{\text{Acc}}$ the accepting
state, and all orthogonal states rejecting.

Finally, we will write
$\overline{V}_k = (\bra{\text{Acc}}_{\reg{F}}\otimes I_{\reg{M}})V_k$
to denote the map which applies the verification, but only outputs the
state corresponding to the verification procedure accepting, corresponding to the procedure $\tsf{QAS.Ver}_k'$ described in \cref{sc:qas}. Then note that $\overline{V}_k=A_k^\dagger$.

From this authentication scheme and fixed state $\ket{\psi}$, which can be assumed without loss of generality to be $\ket{0}$, we
construct a copy protection scheme for point functions of length $n$,
$\tsf{AuthCP}$, as follows:

\begin{description}
\item[\rm$\mathsf{AuthCP.Protect}(p)$:]
	On input $p\in\{0,1\}^n$, do the following:
	\begin{enumerate}
		\item
			Output $A_p \ket{\psi}$.
	\end{enumerate}
\item[\rm$\mathsf{AuthCP.Eval}(\sigma,x)$:]
	On input $\sigma\in{\cal D}(\reg{Y})$ and $x\in\{0,1\}^n$, do the
	following:
	\begin{enumerate}
		\item
			Compute $\xi=V_{x} \sigma V_{x}^\dagger$. Recall that
			$\xi$ is a state on registers $\sf F$, the flag register, and
			$\sf M$, the message register.
		\item
			Measure the $\sf F$ register of $\xi$ in
			$\{\ketbra{\text{Acc}},I-\ketbra{\text{Acc}}\}$. If the
			outcome obtained is ``Acc'', output $1$. Otherwise, output $0$.
	\end{enumerate}
\end{description}

We recall that correctness is parametrized by a family of input distributions to each point function, and security is parametrized
by a distribution on the possible functions to be encoded and by a family of
distributions on challenges to send the users Bob and Charlie.
Our correctness and security are proven with respect to the
following distributions:
\begin{itemize}
\item Our correctness will be with respect to the distribution $\Dhalf_p$, as defined in \cref{df:Dy}, which we recall is the distribution on $\{0,1\}^n$ in which $p$ is sampled with probability $1/2$, and all other strings are sampled with probability $\frac{1}{2(2^n-1)}$.
	\item
		In our security proof, we will assume that the point $p$ of the challenge function is chosen uniformly
		at random. This corresponds to the distribution $R$ given in \cref{df:Dn}.
	\item
		If the challenge function is specified by the point $p$, the
		challenges will be sampled independently according to the distribution $\Dhalf_p$. We will refer to this as $\Dhalf_p\times\Dhalf_p$.
\end{itemize}
We first prove the correctness of the scheme $\sf AuthCP$.

\begin{theorem}
If the scheme $\QAS$ is an $\epsilon$-total authentication scheme, then the
scheme $\sf AuthCP$ described above is $\epsilon$-correct with respect to the family of
distributions $\Dhalffam$.
\end{theorem}
\begin{proof}
For all $p \in \{0,1\}^n$, it suffices to compute a lower bound on
\begin{equation}
\label{eq:correctness-to-bound}
	\frac{1}{2}
	\norm{\overline{V}_p A_p\ket{\psi}}^2
	+
	\frac{1}{2} \cdot \frac{1}{2^n-1}
	\sum_{\substack{x \in \{0,1\}^n\\ x \not= p}}
	\left(
		1
		-
		\norm{
			\overline{V}_x A_p \ket{\psi}
		}^2
	\right).
\end{equation}
By the correctness of the authentication scheme, we have that
$\norm{\overline{V}_p A_p \ket{\psi}}^2=1$.
On the other hand, by \cref{th:QAS-wrong-key}, we have that
\begin{equation}
	\sum_{\substack{x \in \{0,1\}^n \\ x \not= p}}
		\norm{
			\overline{V}_x A_p\ket{\psi}
		}^2
	\leq
	2^n\cdot 2\epsilon - 1
\end{equation}
by expanding the expectation and removing the term corresponding to
$x = p$. Thus, a lower bound for \cref{eq:correctness-to-bound} is given
by
\begin{equation}
	\frac{1}{2}
	+
	\frac{1}{2} \cdot \frac{1}{2^n-1}
	\left(
		2^n-1
		-
		\sum_{\substack{x \in \{0,1\}^n \\ x \not= p}}
		\norm{
			\overline{V}_x A_p\ket{\psi}
		}^2
	\right)
	\geq
	\frac{1}{2}
	+
	\frac{1}{2}
	\left(1 - \frac{2^n\cdot 2 \epsilon - 1}{2^n-1}\right)
	\geq
	1 - \epsilon,
\end{equation}
as long as $\epsilon\leq 1/2$, and so the scheme is $\epsilon$-correct with respect to the given
distribution family.\end{proof}

%~~~~~~~~~~~~~~~~~~~~~~~~~~~~~~~~~~~~~~%
\subsection{Honest-Malicious Security} %
%~~~~~~~~~~~~~~~~~~~~~~~~~~~~~~~~~~~~~~%

In this section, we prove the security of the scheme $\tsf{AuthCP}$
in the honest-malicious setting.
Formally, we prove the following theorem.

\begin{theorem}\label{thm:hm-security-auth}
If the scheme $\QAS$ is an $\epsilon$-total authentication scheme, then the scheme $\sf AuthCP$ described above is
$(\frac{3}{2}\epsilon+\sqrt{2\epsilon})$-honest-malicious secure with
respect to the uniform distribution $R$ on point functions and
challenge distributions $\{\Dhalf_p\times\Dhalf_p\}_{p \in \{0,1\}^n}$, where $R$ and
$\Dhalf_p$ are as defined in \cref{df:Dn} and \cref{df:Dy}.
\end{theorem}

In fact, we can prove security with respect to a slightly more general set of challenge distributions. If we let $T^{(r)}_p$ be the distribution that samples $p$ with probability $r$, and any other point uniformly, then for any $r\in [1/2,1]$, our proof holds when Bob's input is chosen according to $T^{(r)}_p$ and Charlie's input is chosen according to $\Dhalf_p$. (See \cref{rem:gendist} following the proof of \cref{thm:hm-security-auth}).
If Bob gets the point with probability less than $1/2$, then it becomes easier for the adversary to win. Pete can simply send the program to Charlie, and give Bob a maximally mixed state. In that case, Bob will probably output 0, which is correct more than $1/2$ the time.

For the challenge distributions $R$ and $\{\Dhalf_p\times\Dhalf_p\}_p$, it is easy to see that Charlie's maximum guessing probability if he has no interaction with Pete, against which we measure security (see \cref{defn:copy-protection}), is $\pmarg=1/2$. We will use this fact in our security proof, which could likely be generalized to other distributions of Charlie's challenge with a different value of $\pmarg$, but we do not analyze such cases.

\begin{figure}
\begin{center}
\begin{tikzpicture}
\draw[double] (0,1.5) -- (3.625,1.5) -- (3.625,1);
\draw (0,.625) -- (4,.625);		\draw (4,.825) -- (5,.825);	\draw[double] (5,.825) -- (6,.825);
									\draw (4,.425) -- (4.3,.425);
\draw (1.1,-.625) -- (4,-.625);		\draw[double] (4,-.625) -- (6,-.625);
\draw[double] (0,-1.5) -- (3.625,-1.5) -- (3.625,-1);

\node at (-.25,1.5) {$x_1$};
\node at (-.55,.625) {$A_p\ket{\psi}$};
\node at (.8,-.625) {$\ket{0}$};
\node at (-.25,-1.5) {$x_2$};

\node at (1.3,.8) {$\reg{Y}$};
\node at (1.3,-.45) {$\reg{Z}$};
\filldraw[fill=white] (1.5,1) rectangle (2.25,-1);	\node at (1.875,0) {$U_{\pirate}$};
\node at (2.45,.8) {$\reg{Y}$};
\node at (2.45,-.45) {$\reg{Z}$};
\draw[dashed] (.25,1.25) rectangle (2.65,-1.25);	\node at (.65,1.05) {\small Pete};

\filldraw[fill=white] (3.25,1) rectangle (4,.25);		\node at (3.625,.625) {$V_{x_1}$};
\node at (4.2,1) {\small$\reg{F}$};
\node at (4.2,.6) {\small$\reg{M}$};
\node at (5.2,.825) {\meas};
\draw[dashed] (3,1.75) rectangle (5.7,.1);			\node at (5.3,1.55) {\small Bob};

\filldraw[fill=white] (3.25,-1) rectangle (5.55,-.25);	\node at (4.4,-.625) {\small$\{\Pi_{x_2},I-\Pi_{x_2}\}$};
\draw[dashed] (3,-1.75) rectangle (5.7,-.1);			\node at (5.05,-1.55) {\small Charlie};

\node at (6.25,.825) {$b_1$};
\node at (6.25,-.625) {$b_2$};

\end{tikzpicture}
\end{center}
\caption{%
\label{fig:auth-security}%
The pirating game specified to the
$\tsf{AuthCP}$ scheme.}
\end{figure}
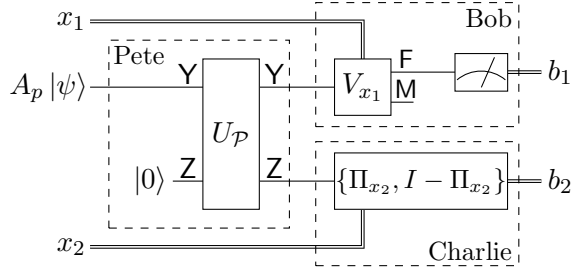

The idea of the proof is the following. In the setting of the scheme
$\sf AuthCP$, the pirating game $\ExperimentFree_{\hatAdv,{\sf AuthCP}}$
(see \cref{fig:freeloading-experiment}) that an honest-malicious
adversary must win is expressed in \cref{fig:auth-security}. Without
loss of generality, we can assume Pete's behaviour is modeled by a
unitary $U_{\pirate}$ on the space $\reg{YZ}$ for an arbitrary auxiliary
space $\reg{Z}$ initialized to a fixed state, which we will denote
$\ket{0}$. (Note that this state can be composed of more than one
qubit.)

Since the adversary is honest-malicious, we can assume that Bob is
honestly evaluating the program, meaning he runs the verification
procedure of the underlying authentication scheme, using the point he
receives as the key, on the register $\reg{Y}$, outputting 1 if and only
if the flag register~$\reg{F}$ is measured as ``Acc''.

Charlie's behaviour can be arbitrary, but without loss of generality, we
can assume that it is specified by a family of two-outcome measurements
on $\reg{Z}$, $\{\Pi_x,I-\Pi_x\}_{x\in\{0,1\}^n}$. Charlie uses his
challenge input $x_2$ to select a measurement to perform to obtain his
output $b_2$.

We will break the proof into two cases. First, consider the case where
$x_1=p$. We can consider Pete's output in two orthogonal parts:
\begin{equation}
	U_{\pirate}(A_{p}\ket{\psi}\otimes\ket{0})
	=
	\ket{\Gamma_{\text{Acc}}^{p}}+\ket{\Gamma_{\text{Rej}}^{p}},
\end{equation}
where $\ket{\Gamma_{\text{Acc}}^p}$ is the part of the state that
leads to Bob outputting $1$ on input $p$, which is the correct bit for Bob
to produce in this case. That is, $\ket{\Gamma_{\text{Acc}}^p}$ is the
projection of Pete's output onto states where the $\reg{Y}$ register is
supported on the image of $A_{p}$. When $x_1=p$, only
$\ket{\Gamma_{\text{Acc}}^p}$ contributes to a winning outcome. We show
(\cref{lem:Bob-correct}) that this state is close (on average over $p$)
to a state of the form $A_{p}\ketbra{\psi}A_p^\dagger\otimes \xi_{\reg{Z}}$ for some subnormalized state
$\xi$ independent of $p$. Since Charlie's input is essentially
independent of $p$, his winning probability is not much more than $1/2$,
so the total winning probability in this case is not much more than
$1/2$ (\cref{lem:independent-state}), which is scaled down by the trace
of the subnormalized state $\xi$, representing the fact that the
probability that Bob outputs the correct bit is
$\norm{\ket{\Gamma_{\text{Acc}}^p}}^2$.

The other case is when $x_1 \neq p$. In that case, we need to consider
the contribution of both terms $\ket{\Gamma_{\text{Acc}}^p}$ and
$\ket{\Gamma_{\text{Rej}}^p}$, as well as their cross term. We can bound
the contribution of the first term to just over
$\frac{1}{2}\Tr(\xi)$ because Charlie's input is close to
$p$-independent. As for the contribution of the second term, in the
worst case, the second term is of the form
$\alpha\ket{0}_{\reg{Y}}\otimes A_{p}\ket{\psi}$ for some scaling factor~$\alpha$.
This corresponds to the strategy that Pete just sends Charlie the
program. Charlie can evaluate the program and be correct with
probability close to 1, and Bob will output 0 with probability close to
1, which is the correct bit in this case, since $x_1\neq p$. So we
trivially upper bound the contribution of this term by
$\norm{\ket{\Gamma_{\text{Rej}}^p}}^2$. However, as this increases,
the size of $\ket{\Gamma_{\text{Acc}}^p}$ and thus $\Tr(\xi)$
decreases, so the probability of being correct in the $x_1=p$ case goes
down. We find that the total contribution, ignoring the cross term, is
at most negligibly more than $1/2$. Finally, we show that the cross-term
is negligible by the correctness of the scheme $\sf AuthCP$.

We first state and prove the necessary lemmas, before formalizing the
above argument.
The following lemma is simply stating that if Charlie gets an input
state that is independent of the point $p$, then his guess as to whether
$x_2=p$ will be independent of $p$, and so will be correct with
probability $1/2$.

\begin{lemma}\label{lem:independent-state}
Suppose $p$ is chosen uniformly at random, and $x_2 \leftarrow \Dhalf_p$, so that
with probability $1/2$, $x_2=p$, and otherwise $x_2$ is
uniform on $\{0,1\}^n\setminus\{p\}$.  Let $\Pi_{x_2}^1=\Pi_{x_2}$ and $\Pi_{x_2}^0=I-\Pi_{x_2}$, so $\Pi_{x_2}^{P_p(x_2)}=\Pi_{x_2}$ when
$x_2=p$, and otherwise $\Pi_{x_2}^{P_p(x_2)} = I-\Pi_{x_2}$. Then, for any
density matrix $\sigma$,
$\E_{p,x_2}\Tr\left(\Pi_{x_2}^{P_p(x_2)}\sigma\right)=\frac{1}{2}$.
\end{lemma}
\begin{proof}
It suffices to compute:
\begin{align*}
\E_{p,x_2}\Tr(\Pi_{x_2}^{P_p(x_2)}\sigma)
&= \frac{1}{2^n}\sum_{p\in \{0,1\}^n}\left(\frac{1}{2}\Tr(\Pi_p\sigma)+\frac{1}{2}\frac{1}{2^n-1}\sum_{x_2\neq p}\Tr((I-\Pi_{x_2})\sigma)\right)\\
&= \frac{1}{2}\frac{1}{2^n}\sum_{p\in \{0,1\}^n}\Tr(\Pi_p\sigma)+\frac{1}{2}\left(1-\frac{1}{2^n}\sum_{x_2\in \{0,1\}^n}\Tr(\Pi_{x_2}\sigma)\right)=\frac{1}{2}.\qedhere
\end{align*}
\end{proof}

The following lemma tells us that in the part of Pete's output that will be accepted by Bob in the $x_1=p$ case, Bob's input from Pete is essentially $A_p\ket{\psi}$, and Charlie's input from Pete is almost independent of $p$. Recall that $A_p$ is an isometry, so $A_pA_p^\dagger$ is the projector onto $\text{im}(A_p)$.

\begin{lemma}\label{lem:Bob-correct}
Let $
	\ket{\Gamma_{\text{Acc}}^p}_{\reg{YZ}}
	=
	(A_{p}A_{p}^\dagger \otimes I_{\reg{Z}})
	U_{\pirate}(A_{p}\ket{\psi}\otimes\ket{0})
$ be the projection of Pete's output onto states supported on
$\text{im}(A_{p})$ in the $\reg{Y}$ register. Then, there exists a
subnormalized state $\xi \in \mc{D}(\reg{Z})$ such that
\begin{equation}
	\E_p\Delta\left(
		\ketbra{\Gamma_{\text{Acc}}^p}
		,
		A_{p}\ketbra{\psi}A_{p}^\dagger
		\otimes
		\xi
	\right)
	\leq
	\epsilon.
\end{equation}
\end{lemma}
\begin{proof}
By the security of the total authentication scheme, there exists a
completely positive trace non-increasing map
$\Psi : \mc{L}(\reg{Z}) \to \mc{L}(\reg{Z})$ such that
\begin{equation}
\begin{split}
	&
	\E_p
		\ketbra{p}
		\otimes
		(\overline{V}_p \otimes I_{\reg{Z}})(U_{\pirate})_{\reg{YZ}}
			(A_p\ketbra{\psi} A_p^\dagger\otimes\ketbra{0}_{\reg{Z}})
		(U_{\pirate})_{\reg{YZ}}^\dagger
		(\overline{V}_p^\dagger\otimes I_{\reg{Z}})
	\\\approx_{\epsilon}&
	\E_p
		\ketbra{p}
		\otimes
		(\overline{V}_pA_p)
			\ketbra{\psi}_{\reg{M}}
		(A_p^\dagger\overline{V}_p^\dagger)
		\tensor
		\Psi(\ketbra{0})
		.
\end{split}\label{eq:exp-over-p}
\end{equation}

Using the fact that $\overline{V}_p=A_p^\dagger = A_p^\dagger( A_pA_p^\dagger)$ (that is, project onto states in the image of $A_p$, and then invert $A_p$), we have
\begin{align*}
(\overline{V}_p\otimes I_{\reg{Z}})U_{\pirate} (A_{p}\ket{\psi}\otimes \ket{0}_{\reg{Z}})
&= (\overline{V}_p\otimes I_{\reg{Z}})(A_{p}A_{p}^\dagger\otimes
I_{\reg{Z}})U_{\pirate} (A_{p}\ket{\psi}\otimes \ket{0}_{\reg{Z}})\\
&= (\overline{V}_p\otimes I_{\reg{Z}})\ket{\Gamma_{\text{Acc}}^{p}}.
\end{align*}
Then by \cref{th:trace-distance-orthogonal}, and letting
$\xi=\Psi(\ketbra{0})$, we can continue from \cref{eq:exp-over-p} to get:
\begin{align*}
\E_p\Delta\left( \overline{V}_{p}\ketbra{\Gamma_{\text{Acc}}^{p}}\overline{V}_{p}^\dagger , \overline{V}_{p}A_{p}\ketbra{\psi}A_{p}^\dagger \overline{V}_{p}^\dagger \otimes \xi \right)   &\leq\epsilon\\
\E_p\Delta\left( \ketbra{\Gamma_{\text{Acc}}^{p}}, A_{p}\ketbra{\psi}A_{p}^\dagger \otimes \xi \right)   &\leq\epsilon,
\end{align*}
where we used the fact that $\ket{\Gamma_{\text{Acc}}^{p}}$ and $A_{p}\ket{\psi}$ are both orthogonal to the kernel of $\overline{V}_{p}$.
\end{proof}

We now proceed to prove our main theorem of this section, \cref{thm:hm-security-auth}.

\begin{proof}[Proof of \cref{thm:hm-security-auth}]
For a fixed $p$, $x_1$ and $x_2$, let $q_1^{p,x_2}$ be the adversary's winning
probability when $x_1=p$, and let $q_0^{p,x_1,x_2}$ be the winning
probability when $x_1\neq p$. Then the total winning probability is
given by
\begin{equation}
	\frac{1}{2}\E_{\substack{p\leftarrow R,\\x_2\leftarrow \Dhalf_p}}q_1^{p,x_2}
	+
	\frac{1}{2}\E_{\substack{p\leftarrow R,\\ x_1\leftarrow \{0,1\}^n\setminus p,\\x_2\leftarrow\Dhalf_p}}q_0^{p,x_1,x_2}.
\end{equation}
If $\ket{\Gamma^p}:=U_{\pirate}(A_p\ket{\psi}\otimes\ket{0})$ is Pete's
output for a fixed $p$, and $\Pi_{x_2}^{P_p(x_2)}$ is defined to be $\Pi_p$ when
$x_2=p$ and $I-\Pi_{x_2}$ otherwise, we have that
\begin{equation}
\begin{split}
	q_1^{p,x_2}
	&=
	\norm{
		(
			\ketbra{\text{Acc}}_{\reg{F}}
			\otimes
			I_{\reg{M}}
			\otimes
			(\Pi_{x_2}^{P_p(x_2)})_{\reg{Z}}
		)
		(V_{p}\otimes\Id_{\reg{Z}})
		\ket{\Gamma^{p}}
	}^2
	\\\qq{and}
	q_0^{p,x_1,x_2}
	&=
	\norm{
		(
			(I_{\reg{F}}-\ketbra{\text{Acc}}_{\reg{F}})
			\otimes
			I_{\reg{M}}
			\otimes
			(\Pi_{x_2}^{P_p(x_2)})_{\reg{Z}}
		)
		(V_{x_1}\otimes\Id_{\reg{Z}})
		\ket{\Gamma^{p}}
	}^2.
\end{split}
\end{equation}
We will upper bound $q_1^{p,x_2}$ and $q_0^{p,x_1,x_2}$ separately.

Recall that we can write Pete's output as
\begin{equation}
	\ket{\Gamma^p}
	=
	\ket{\Gamma_{\text{Acc}}^p}
	+
	\ket{\Gamma_{\text{Rej}}^p},
\end{equation}
where
\begin{equation}
\begin{split}
	&
	\ket{\Gamma_{\text{Acc}}^p}
	=
	(A_{p}A_{p}^\dagger \otimes I_{\reg{Z}})
	\ket{\Gamma^{p}}
	\\\qq{and}&
	\ket{\Gamma_{\text{Rej}}^p}
	=
	((I_{\reg{Y}}-A_{p}A_{p}^\dagger) \otimes I_{\reg{Z}})
	\ket{\Gamma^p}.
\end{split}
\end{equation}

\paragraph{The $x_1=p$ case.}
We begin by upper bounding $q_1^{p,x_2}$. We first show there is no
contribution from the second term:
\begin{equation}
\begin{split}
	&
	(
		\ketbra{\text{Acc}}_{\reg{F}}
		\otimes
		I_{\reg{M}}
		\otimes
		\Pi_{x_2}^{P_p(x_2)}
	)
	(V_{p} \otimes I_{\reg{Z}})
	\ket{\Gamma^{p}_{\text{Rej}}}
	\\
	={}&
	(
		\ketbra{\text{Acc}}_{\reg{F}}
		\otimes
		I_{\reg{M}}
		\otimes
		(\Pi_{x_2}^{P_p(x_2)})_{\reg{Z}}
	)
	(
		V_{p}
		(I_{\reg{Y}}-A_{p}A_{p}^\dagger)
		\otimes
		I_{\reg{Z}}
	)
	\ket{\Gamma^p}
	\\
	={}&
	0
\end{split}
\end{equation}
because
\begin{equation}
\begin{split}
	(\bra{\text{Acc}}\otimes I_{\reg{M}})
	V_{p}
	(I_{\reg{Y}}-A_{p}A_{p}^\dagger)
	&=
	\overline{V}_{p}(I_{\reg{Y}}-A_{p}A_{p}^\dagger)
	\\&=
	A_{p}^\dagger A_{p}A_{p}^\dagger
	(I_{\reg{Y}}-A_{p}A_{p}^\dagger)
	\\&=
	0.
\end{split}
\end{equation}
Above we used the fact that $\overline{V}_p=A_p^\dagger=A_p^\dagger
(A_pA_p^\dagger)$ which is to say that $\overline{V}_p$ simply projects
onto states in the image of $A_p$, and then inverts $A_p$. Thus
(omitting implicit tensored identities):
\begin{equation}
\begin{split}
	q_1^{p,x_2}
	&=
	\Tr\left(
		(\ketbra{\text{Acc}}_F\otimes \Pi_{x_2}^{P_p(x_2)})
		V_{p}
		\ketbra{\Gamma_{\text{Acc}}^p}
		V_{p}^\dagger
	\right)
	\\&=
	\Tr\left(
		V_{p}^\dagger
		(\ketbra{\text{Acc}}_F \otimes \Pi_{x_2}^{P_p(x_2)})
		V_{p}
		(
			A_{p}\ketbra{0}A_{p}^\dagger
			\otimes
			\xi
			+
			\delta_p)
		\right)
	\\&\leq
	\Tr\left(
		V_{p}^\dagger
		\ketbra{\text{Acc}}V_{p}A_{p}
		\ketbra{0}A_p^\dagger
		\otimes
		\Pi_{x_2}^{P_p(x_2)}
		\xi
	\right)
	\\&\qquad
	+\Delta\left(
		\ketbra{\Gamma_{\text{Acc}}^p}
		,
		A_{p}\ketbra{0}A_{p}^\dagger\otimes \xi
	\right)
	\\&=
	\Tr(\Pi_{x_2}^{P_p(x_2)}\xi)
	+
	\Delta\left(
		\ketbra{\Gamma_{\text{Acc}}^{p}}
		,
		A_{p}\ketbra{0}A_{p}^\dagger\otimes \xi
	\right),
\end{split}
\end{equation}
where $
	\delta_{p}
	=
	\ketbra{\Gamma_{\text{Acc}}^{p}}
	-
	A_{p}\ketbra{0}A_{p}^\dagger\otimes \xi
$.

By \cref{lem:independent-state}, we have $
	\Tr(\xi)
	\E_{p,x_2}\Tr\left(
		\Pi_{x_2}^{P_p(x_2)}\frac{\xi}{\Tr(\xi)}
	\right)
	=
	\Tr(\xi)/2
$. Combining this with \cref{lem:Bob-correct}, we conclude with:
\begin{equation}
\label{eq:case-p-equal-p0}
	\E_{p,x_2}
	q_1^{p,x_2}
	\leq
	\frac{\Tr(\xi)}{2}
	+
	\epsilon
	.
\end{equation}

\paragraph{The $x_1\neq p$ case:} We will analyze the probability in
three parts, as follows:
\begin{equation}
\label{eq:q0xpp1}
\begin{split}
	q_0^{p,x_1,x_2}
	&
	=
	\norm{
		((I_{\reg{F}}-\ketbra{\text{Acc}}_{\reg{F}})\otimes (\Pi_{x_2}^{P_p(x_2)})_{\reg{Z}})
		V_{x_1}
		\ket{\Gamma^p}
	}^2
	\\&\leq
	\underbrace{
		\norm{
			((I_{\reg{F}}-\ketbra{\text{Acc}}_{\reg{F}})\otimes \Pi_{x_2}^{P_p(x_2)})
			V_{x_1}
			\ket{\Gamma_{\text{Acc}}^{p}}
		}^2
	}_{=T_1^{p,x_1,x_2}}
	+
	\underbrace{
		\norm{
			((I_{\reg{F}}-\ketbra{\text{Acc}}_{\reg{F}}) \otimes \Pi_{x_2}^{P_p(x_2)})
			V_{x_1}
			\ket{\Gamma_{\text{Rej}}^p}
		}^2
	}_{=T_2^{p,x_1,x_2}}
	\\&\qquad+
	\underbrace{
		2\abs{
			\bra{\Gamma_{\text{Rej}}^p}
			(
				V_{x_1}^\dagger
				(I_{\reg{F}}-\ketbra{\text{Acc}}_{\reg{F}})
				V_{x_1}
				\otimes
				\Pi_{x_2}^{P_p(x_2)}
			)
			\ket{\Gamma_{\text{Acc}}^{p}}
		}
	}_{=T_{\text{cross}}^{p,x_1,x_2}}.
\end{split}
\end{equation}

We begin with the first term, whose analysis is similar to the $x_1=p$
case. We have:
\begin{equation}
\begin{split}
	T_1^{p,x_1,x_2}
	={}&
	\Tr\left(
		(
			V_{x_1}^\dagger
			(I-\ketbra{\text{Acc}})
			V_{x_1}
			\otimes
			\Pi_{x_2}^{P_p(x_2)}
		)
		(
			A_{p}
			\ketbra{0}A_{p}^\dagger
			\otimes
			\xi
			+
			\delta_{p}
		)
	\right)
	\\\leq{}&
	\Tr(\Pi_{x_2}^{P_p(x_2)}\xi)
	+
	\Delta(
		\ketbra{\Gamma_{\text{Acc}}^{p}}
		,
		A_{p}\ketbra{0}A_{p}^\dagger\otimes\xi
	)
	.
\end{split}
\end{equation}
Thus, just as we concluded with \cref{eq:case-p-equal-p0}, we can
conclude
\begin{equation}
\label{eq:term1}
	\E_{p,x_1,x_2} T_1^{p,x_1,x_2}
	\leq
	\frac{\Tr(\xi)}{2}+\epsilon
	,
\end{equation}
again, by \cref{lem:independent-state} and \cref{lem:Bob-correct}.

For the second term, we will use the naive bound:
\begin{equation}
\begin{split}
	T_2^{p,x_1,x_2}
	&\leq
	\norm{\ket{\Gamma_{\text{Rej}}^{p}}}^2
	\\&=
	1-\norm{\ket{\Gamma_{\text{Acc}}^{p}}}^2
	\\&\leq
	1
	-
	\Tr(
		A_{p}\ketbra{0}A_{p}^\dagger\otimes \xi
	)
	+
	\Delta\left(
		\ketbra{\Gamma_{\text{Acc}}^{p}}
		,
		A_{p}\ketbra{0}A_{p}^\dagger\otimes \xi
	\right)
	\\&=
	1
	-
	\Tr(\xi)
	+
	\Delta\left(
		\ketbra{\Gamma_{\text{Acc}}^{p}}
		,
		A_{p}\ketbra{0}A_{p}^\dagger
		\otimes
		\xi
	\right)
	.
\end{split}
\end{equation}
Then by \cref{lem:Bob-correct}, we have
\begin{align}
\label{eq:term2}
	\E_{p,x_1,x_2}T_2^{p,x_1,x_2}
	\leq
	1-\Tr(\xi)+\epsilon
	.
\end{align}

Finally, we upper bound the cross-term. The idea is that $\ket{\Gamma_{\text{Acc}}^{p}}$ and $\ket{\Gamma_{\text{Rej}}^{p}}$ are orthogonal in the $\reg{Y}$ register. This is, of course, also true once we apply $\Pi_{x_2}^{P_p(x_2)}$ to the $\reg{Z}$ register. Applying the projector $V_{x_1}^\dagger(I_{\reg{F}}-\ketbra{\text{Acc}}_{\reg{F}})V_{x_1}$ to the $\reg{Y}$ register could change this, however, we will argue that, by correctness of the scheme, this projector cannot change the state $\ket{\Gamma_{\text{Acc}}^{p}}$ very much, because its first register is in $\text{im}(A_{p})$, and trying to decode with a different key, $x_1\neq p$, should result in rejection with high probability.
We have:
\begin{equation}
\begin{split}
	T_{\text{cross}}^{p,x_1,x_2}
	&=
	2\abs{
		\bra{\Gamma_{\text{Rej}}^{p}}
		(I_{\reg{Y}}\otimes\Pi_{x_2}^{P_p(x_2)})
		\ket{\Gamma_{\text{Acc}}^{p}}
		-
		\bra{\Gamma_{\text{Rej}}^{p}}
		(
			V_{x_1}^\dagger
			\ketbra{\text{Acc}}
			V_{x_1}
			\otimes
			\Pi_{x_2}^{P_p(x_2)}
		)
		\ket{\Gamma_{\text{Acc}}^{p}}
	}
	\\&=
	2\abs{
		\bra{\Gamma_{\text{Rej}}^{p}}
		(
			V_{x_1}^\dagger
			\ketbra{\text{Acc}}
			V_{x_1}
			\otimes\Pi_{x_2}^{P_p(x_2)}
		)
		\ket{\Gamma_{\text{Acc}}^{p}}
	}\\
	&\leq
	2\norm{
		(\bra{\text{Acc}}V_{x_1}\otimes I_{\reg{Z}})
		\ket{\Gamma_{\text{Acc}}^{p}}
	}
=2\norm{(\overline{V}_{x_1}\otimes I_{\reg{Z}})\ket{\Gamma_{\text{Acc}}^{p}}}
	,
\end{split}
\end{equation}
where, in the last line, we used the Cauchy-Schwarz inequality. Since $\ket{\Gamma_{\text{Acc}}^p}$ is supported on $\text{im}(A_p)$ in the first register, it has a Schmidt decomposition of the form:
\begin{equation}
	\ket{\Gamma_{\text{Acc}}^{p}}
	=
	\sum_{\ell}
		\beta_{\ell}
		(A_{p}\ket{u_\ell})_{\reg{Y}}
		\otimes
		\ket{v_{\ell}}_{\reg{Z}}.
\end{equation}
Taking the expectation, we have:
\begin{equation}
\label{eq:cross-term}
\begin{split}
	\E_{p,x_1,x_2}
	T_{\text{cross}}^{p,x_1,x_2}
	&\leq
	2\E_{p,x_1}
		\sqrt{
			\sum_{\ell}
				\abs{\beta_\ell}^2
				\norm{
					\overline{V}_{x_1}A_{p}
					\ket{u_{\ell}}
				}^2
		}
	\\&\leq
	2\sqrt{
		\sum_{\ell}
			\abs{\beta_\ell}^2
			\E_{p,x_1}
				\norm{
					\overline{V}_{x_1}
					A_{p}
					\ket{u_{\ell}}
				}^2
	}
	\qq{(By Jensen's inequality).}
\end{split}
\end{equation}
We next want to appeal to \cref{th:QAS-wrong-key},
which implies that
$
	\E_{p,x_1\leftarrow\{0,1\}^n}\norm{\overline{V}_{x_1}A_{p}\ket{u}}^2
	\leq
	2\epsilon
$, for any pure state~$\ket{u}$, however, notice that $p$ and $x_1$ are not uniformly distributed, because while $p$ is uniform, $x_1$ is uniform over $\{0,1\}^n\setminus \{p\}$. However, since for any $p$ we have $\norm{\overline{V}_pA_p\ket{u}}^2=1$, we have:
\begin{align*}
\E_{\substack{p\leftarrow \{0,1\}^n,\\x_1\leftarrow\{0,1\}^n\setminus\{p\}}}\norm{\overline{V}_{x_1}A_p\ket{u_\ell}}^2
&= \frac{2^{2n}}{2^n(2^n-1)}\left(\E_{\substack{p\leftarrow \{0,1\}^n,\\x_1\leftarrow\{0,1\}^n}}\norm{\overline{V}_{x_1}A_p\ket{u_{\ell}}}^2-\frac{1}{2^{2n}}\sum_{p\in\{0,1\}^n}\norm{\overline{V}_{p}A_p\ket{u_{\ell}}}^2\right)\\
&\leq 2\epsilon + \frac{1}{2^n-1}2\epsilon-\frac{1}{2^n-1}
\end{align*}
which is at most $2\epsilon$ as long as $\epsilon \leq 1/2$.
Thus we can continue:
\begin{equation}
\begin{split}
	\E_{p,x_1,x_2} T_{\text{cross}}^{p,x_1,x_2}&\leq
	2\sqrt{\sum_{\ell}\abs{\beta_\ell}^2}\sqrt{2\epsilon}
	\\&=
	2\sqrt{2\epsilon}
	.
\end{split}
\end{equation}

Combining \cref{eq:term1}, \cref{eq:term2}, and \cref{eq:cross-term} into \cref{eq:q0xpp1}, we conclude the $x_1\neq p$ case with:
\begin{equation}
\label{eq:case-p-nequal-p0}
\begin{split}
	\E_{p,x_1,x_2}
	q_0^{p,x_1,x_2}
	&\leq
	\E_{p,x_1,x_2}
	T_1^{p,x_1,x_2}
	+
	\E_{p,x_1,x_2}
	T_2^{p,x_1,x_2}
	+
	\E_{p,x_1,x_2}
	T_{\text{cross}}^{p,x_1,x_2}
	\\&\leq
	\frac{1}{2}
	\Tr(\xi)
	+
	\epsilon
	+
	1
	-
	\Tr(\xi)
	+
	\epsilon
	+
	2\sqrt{2\epsilon}
	.
\end{split}
\end{equation}

\paragraph{Conclusion.} We can now combine \cref{eq:case-p-equal-p0} and
\cref{eq:case-p-nequal-p0} to get an upper bound on the total winning
probability of:
\begin{equation}
\begin{split}
	&
	\frac{1}{2}
	\E_{p,x_2}
	q_1^{p,x_2}
	+
	\frac{1}{2}
	\E_{p,x_1,x_2}
	q_0^{p,x_1,x_2}
	\\\leq{}&
	\frac{1}{2}\left(
		\frac{1}{2}
		\Tr(\xi)
		+
		\epsilon
	\right)
	+
	\frac{1}{2}\left(
		1
		-
		\frac{1}{2}\Tr(\xi)
		+
		2\epsilon
		+
		2\sqrt{2\epsilon}
	\right)
	\\={}&
	\frac{1}{2}+\frac{3}{2}\epsilon+\sqrt{2\epsilon}
	.
\end{split}\label{eq:conclusion}
\end{equation}
Noting that $p^{\text{marg}}_{R, \{\Dhalf_p\times\Dhalf_p\}_p} = \frac{1}{2}$
completes the proof.
\end{proof}

\begin{remark}\label{rem:gendist}
We note that if our challenge distribution instead chooses Bob's input so that $x_1=p$ with probability $r$, for $r\geq 1/2$, and all other points uniformly, then \cref{eq:conclusion} would instead give us:
\begin{equation*}
\begin{split}
	&
	r
	\E_{p,x_2}
	q_1^{p,x_2}
	+
	(1-r)
	\E_{p,x_1,x_2}
	q_0^{p,x_1,x_2}
	\\\leq{}&
	r\left(
		\frac{1}{2}
		\Tr(\xi)
		+
		\epsilon
	\right)
	+
	(1-r)\left(
		1
		-
		\frac{1}{2}\Tr(\xi)
		+
		2\epsilon
		+
		2\sqrt{2\epsilon}
	\right)
	\\={}&
	\frac{1}{2}(2r-1)\Tr(\xi)+1-r + (2-r)\epsilon+2(1-r)\sqrt{2\epsilon}\\
\leq{}&\frac{1}{2}(2r-1)+1-r + (2-r)\epsilon+2(1-r)\sqrt{2\epsilon}\\
={}&\frac{1}{2}+(2-r)\epsilon+2(1-r)\sqrt{2\epsilon}.
\end{split}
\end{equation*}
We therefore have $((2-r)\epsilon+2(1-r)\sqrt{2\epsilon})$-honest-malicious security under this more general challenge distribution, where Bob's input is distributed as $T^{(r)}_p$ and Charlie's input is distributed as $\Dhalf_p$.
\end{remark}

%%%%%%%%%%%%%%%%%%%%%%%%%%%%%%%%%%%%%%%%%%%%%%%%
% Bibliography                                 %
% \addcontentsline{toc}{section}{Bibliography} %
%%%%%%%%%%%%%%%%%%%%%%%%%%%%%%%%%%%%%%%%%%%%%%%%

%%%%%%%%%%%
\appendix %
%%%%%%%%%%%

%========================================%
\section{Proofs for {\cref{sc:prelims}}} %
\label{sc:prelim-proofs}                 %
%========================================%

We collect in this appendix the proofs of statements made in
\cref{sc:prelims}.

%~~~~~~~~~~~~~~~~~~~~~~~~~~~~~~~~~~~~~~~~~~~%
\subsection{Proofs for {\cref{sc:trace-distance}}} %
%~~~~~~~~~~~~~~~~~~~~~~~~~~~~~~~~~~~~~~~~~~~%

Before proceeding to the proof of \cref{th:trace-distance-orthogonal},
we recall that for any linear operator $X \in \mc{L}(\tsf{A})$,
$\norm{X}_1 = \max_{U \in \mc{U}(\tsf{A})} \abs{\ip{U}{X}}$. Note that
the absolute value here is superfluous in a certain sense. Specifically,
for any unitary operator $U$ and linear operator $X$, there exists a
unitary operator $V$ such that $\abs{\ip{U}{X}} = \ip{V}{X}$. Indeed,
assume that $\ip{U}{X} = a\cdot e^{i\theta}$ for a non-negative real
$a$. Then, taking $V = e^{i\theta} U$ yields $\ip{V}{X} = a$. Thus, it
would suffice to take the maximum over unitary operators which yield a
real and non-negative value.

\begin{proof}[Proof of {\cref{th:trace-distance-orthogonal}}]
First, note that
\begin{equation}
	\norm{
		\sum_{j \in J} \ketbra{\psi_j} \tensor X_j
		-
		\sum_{j \in J} \ketbra{\psi_j} \tensor Y_j
	}_1
	=
	\norm{
		\sum_{j \in J} \ketbra{\psi_j} \tensor (X_j - Y_j)
	}_1.
\end{equation}

Next, we show that
\begin{equation}
	\norm{
		\sum_{j \in J} \ketbra{\psi_j} \tensor (X_j - Y_j)
	}_1
	\leq
	\sum_{j \in J} \norm{X_j - Y_j}_1.
\end{equation}
To obtain this inequality, it suffices to recall that for any two linear
operators $A$ and $B$, we have that
$\norm{A \tensor B}_1 \leq \norm{A}_1 \cdot \norm{B}_1$. Indeed, by
using the fact that the Schatten-$1$ norm is submultiplicative and
non-increasing under the partial trace \cite{Wat18}, we have that
\begin{equation}
	\norm{A \tensor B}_1
	=
	\norm{(A \tensor I)(I \tensor B)}_1
	\leq
	\norm{A \tensor I}_1 \cdot \norm{I \tensor B}_1
	\leq
	\norm{A}_1\cdot\norm{B}_1.
\end{equation}
Thus, we obtain the desired upper bound by writing
\begin{equation}
	\norm{
		\sum_{j \in J} \ketbra{\psi_j} \tensor (X_j - Y_j)
	}_1
	\leq
	\sum_{j \in J}
	\norm{
		\ketbra{\psi_j} \tensor (X_j - Y_j)
	}_1
	\leq
	\sum_{j \in J}
	\norm{
		\ketbra{\psi_j}
	}_1
	\cdot
	\norm{
		X_j - Y_j
	}_1
\end{equation}
and noting that $\norm{\ketbra{\psi_j}}_1 = 1$.

Finally, we show that
\begin{equation}
	\sum_{j \in J} \norm{X_j - Y_j}_1
	\leq
	\norm{
		\sum_{j \in J} \ketbra{\psi_j} \tensor (X_j - Y_j)
	}_1.
\end{equation}
By our definition of the trace norm, it suffices to find a unitary
operator $U \in \mc{U}(\tsf{AB})$ such that
\begin{equation}
	\frac{1}{2}
	\abs{\ip{U}{\sum_{j \in J} \ketbra{\psi_j}\tensor(X_j - Y_j)}}
	=
	\sum_{j \in J} \norm{X_j - Y_j}_1
\end{equation}
to obtain the inequality. For every $j \in J$, let
$U_J \in \mc{U}(\tsf{B})$ be such that
$\frac{1}{2}\ip{U_j}{X_j - Y_j} = \Delta(X_j,Y_j)$. Note the lack of
absolute value in this equation. Such a unitary $U_j$ must exist by our
remark at the start of this section. It then suffices to take
\begin{equation}
	U
	=
	\left(
		I - \sum_{j \in J} \ketbra{\psi_j}
	\right)
	\tensor
	I
	+
	\sum_{j \in J} \ketbra{\psi_j} \tensor U_j.
\end{equation}
A direct computation then yields
\begin{equation}
	\frac{1}{2}
	\abs{
		\ip{
			U
		}{
			\sum_{j \in J} \ketbra{\psi_j} \tensor (X_j - Y_j)
		}
	}
	=
	\frac{1}{2}
	\abs{
		\sum_{j \in J}
		\ip{U_j}{X_j - Y_j}
	}
	=
	\sum_{j \in J} \Delta(X_j, Y_j)
\end{equation}
which is the desired equality.
\end{proof}

%~~~~~~~~~~~~~~~~~~~~~~~~~~~~~~~~~~~~~~~%
\subsection{Proofs for {\cref{sc:qas}}} %
%~~~~~~~~~~~~~~~~~~~~~~~~~~~~~~~~~~~~~~~%

Before proceeding to the proof of \cref{th:QAS-wrong-key}, we will need
a small lemma which essentially states that orthogonal states are mapped
to orthogonal states by quantum authentication schemes.

\begin{lemma}
\label{th:QAS-orthogonal}
Let $\tsf{QAS}$ be an authentication scheme. Then,
\begin{equation}
	\Delta(\rho, \sigma) = 1
	\implies
	\Delta(\tsf{QAS.Auth}_k(\rho), \tsf{QAS.Auth}_k(\sigma)) = 1
\end{equation}
for all $k \in \mc{K}$ and all states $\rho,\sigma\in\mc{D}(\tsf{M})$.
\end{lemma}

\begin{proof}[Proof of \cref{th:QAS-orthogonal}]
Since the trace distance is contractive under CPTP maps, we have that
\begin{equation}
\begin{split}
	&\Delta(\tsf{QAS.Auth}_k(\rho), \tsf{QAS.Auth}_k(\sigma))
	\\\geq&
	\Delta(
		\Tr_\tsf{F} \circ \tsf{QAS.Ver}_k \circ \tsf{QAS.Auth}_k(\rho),
		\Tr_\tsf{F} \circ \tsf{QAS.Ver}_k \circ \tsf{QAS.Auth}_k(\sigma)
	)
	\\\geq&
	\Delta(
		\rho,
		\sigma
	)
	\\=&
	1.
\end{split}
\end{equation}
Noting that
$\Delta(\tsf{QAS.Auth}_k(\rho),\tsf{QAS.Auth}_k(\sigma)) \leq 1$
completes the proof.
\end{proof}

\begin{proof}[Proof of \cref{th:QAS-wrong-key}]
Let $\tsf{Z} \iso \tsf{Y}$, and define the attack
$\Phi : \mc{L}(\tsf{YZ}) \to \mc{L}(\tsf{YZ})$ by
\begin{equation}
	\xi
	\mapsto
	\text{Swap}_{\tsf{YZ}}\left(
		\Tr_\tsf{Z}(\xi)
		\tensor
		\rho
	\right).
\end{equation}
In other words, an adversary implementing this attack will keep in
memory the authenticated state sent by the sender and will give the
receiver the state $\rho$.

As the authentication scheme is $\epsilon$-total, there exists a
completely positive trace non-increasing map
$\Psi : \mc{L}(\tsf{Z}) \to \mc{L}(\tsf{Z})$ such that
\begin{equation}
	\E_{k \in K}
		\ketbra{k}
		\tensor
		\tsf{QAS.Ver}'_k \circ \Phi \circ \tsf{QAS.Auth}_k(\sigma)
	\approx_\epsilon
	\E_{k \in K}
		\ketbra{k}
		\tensor
		\tsf{QAS.Ver}'_k \circ \Psi \circ \tsf{QAS.Auth}_k(\sigma)
\end{equation}
for all states $\sigma \in \mc{D}(\tsf{MZ})$. In particular, consider
separable states of the form
$\sigma_\tsf{MZ} = \tau_\tsf{M} \tensor \tau'_\tsf{Z}$. For such states,
we have that
\begin{equation}
	\E_{k \in K}
		\ketbra{k}
		\tensor
		\tsf{QAS.Ver}'_k \circ \Phi \circ \tsf{QAS.Auth}_k(\sigma)
	=
	\E_{k \in K}
		\ketbra{k}
		\tensor
		\tsf{QAS.Ver}'_k(\rho)
		\tensor
		\tsf{QAS.Auth}_k(\tau)
\end{equation}
and
\begin{equation}
	\E_{k \in K}
		\ketbra{k}
		\tensor
		\tsf{QAS.Ver}'_k \circ \Psi \circ \tsf{QAS.Auth}_k(\sigma)
	=
	\E_{k \in K}
		\ketbra{k}
		\tensor
		\tau
		\tensor
		\Psi(\tau')
\end{equation}
so the security guarantee yields
\begin{equation}
	\E_{k \in K}
		\ketbra{k}
		\tensor
		\tsf{QAS.Ver}'_k(\rho)
		\tensor
		\tsf{QAS.Auth}_k(\tau)
	\approx_\epsilon
	\E_{k \in K}
		\ketbra{k}
		\tensor
		\tau
		\tensor
		\Psi(\tau').
\end{equation}
Since the trace distance is contractive under CPTP maps, we can trace
out $\tsf{M}$, the message register, to obtain
\begin{equation}
	\E_{k \in K}
		\ketbra{k}
		\tensor
		\Tr\left[
			\tsf{QAS.Ver}'_k(\rho)
		\right]
		\cdot
		\tsf{QAS.Auth}_k(\tau)
	\approx_\epsilon
	\E_{k \in K}
		\ketbra{k}
		\tensor
		\Psi(\tau').
\end{equation}
Using the triangle inequality and two instances of the above equation,
once with $\tau = \ketbra{\psi}$ and once with $\tau = \ketbra{\phi}$
for orthogonal $\ket{\psi}$ and $\ket{\phi}$, we find that
\begin{equation}
\begin{split}
	&
	\E_{k \in \mc{K}}
		\ketbra{k}
		\tensor
		\Tr\left[
			\tsf{QAS.Ver}'_k(\rho)
		\right]
		\cdot
		\tsf{QAS.Auth}_k(\ketbra{\psi})
	\\\approx_{2\epsilon}&
	\E_{k \in K}
		\ketbra{k}
		\tensor
		\Tr\left[
			\tsf{QAS.Ver}'_k(\rho)
		\right]
		\cdot
		\tsf{QAS.Auth}_k(\ketbra{\phi})
\end{split}
\end{equation}
after which we can apply \cref{th:trace-distance-orthogonal} with the
help of the key register, which yields
\begin{equation}
	\E_{k \in K}
		\Tr\left[
			\tsf{QAS.Ver}'_k(\rho)
		\right]
	\cdot
	\Delta\left(
		\tsf{QAS.Auth}_k(\ketbra{\psi})
		,
		\tsf{QAS.Auth}_k(\ketbra{\phi})
	\right)
	\leq 2\epsilon.
\end{equation}
It then suffices to note that
$
	\Delta\left(
		\tsf{QAS.Auth}_k(\ketbra{\psi})
		,
		\tsf{QAS.Auth}_k(\ketbra{\phi})
	\right)
	=
	1
$
by \cref{th:QAS-orthogonal} since~$\ket{\psi}$ is orthogonal
to~$\ket{\phi}$. The desired result follows.
\end{proof}

Now, we proceed to prove \cref{th:QAS-existence}. However, There are a
few technical points which must be covered before. In short, these
points tell us that we can substitute the key set $\mc{K}$ of a QAS with
a set $\mc{K}'$ with little loss of security (\cref{th:QAS-key-change}),
provided that there is map $f : \mc{K}' \to \mc{K}$ which maps a
uniformly random variable on $\mc{K}'$ to an almost uniformly random
variable on $\mc{K}$ (\cref{df:eps}). Our proof is then the application
of these technical arguments to existing theorems concerning unitary
$2$-designs and how they can be used to construct a QAS
(\cref{th:2-design,th:AM17}).

\begin{definition}
\label{df:eps}
A map $f : \mc{A} \to \mc{B}$ between finite sets is $\epsilon$-uniform
if
$
	\frac{1}{2}\sum_{b \in
	\mc{B}}\abs{\frac{\abs{f^{-1}(b)}}{\abs{\mc{A}}} -
	\frac{1}{\abs{\mc{B}}}}
	\leq
	\epsilon.
$
\end{definition}

\begin{lemma}
\label{th:eps-expectation}
Let $f : \mc{A} \to \mc{B}$ be an $\epsilon$-uniform map between finite
sets and $g : \mc{B} \to [0,1]$ be a map. Then,
$\abs{\E_a g(f(a)) - \E_b g(b)} \leq \epsilon$.
\end{lemma}
\begin{proof}
We first note that
\begin{equation}
	\abs{\E_a g(f(a)) - \E_b g(b)}
	=
	\abs{
		\sum_{b \in \mc{B}} \frac{\abs{f^{-1}(b)}}{\abs{\mc{A}}}
		\cdot
		g(b)
		-
		\sum_{b \in \mc{B}} \frac{1}{\abs{\mc{B}}} \cdot g(b)
	}
	\leq
	\abs{
		\sum_{b \in \mc{B}}
		\left(
		\frac{\abs{f^{-1}(b)}}{\abs{\mc{A}}}
		-
		\frac{1}{\abs{\mc{B}}}
		\right)
		g(b)
	}.
\end{equation}
Let $\mc{S} = \{b \in \mc{B} \;|\; \abs{f^{-1}(b)}\abs{\mc{B}} -
\abs{\mc{A}} \geq 0\}$ and $\mc{S}' = \mc{B} \setminus \mc{A}$. We then
have that
\begin{equation}
	\abs{
		\sum_{b \in \mc{B}}
		\left(
		\frac{\abs{f^{-1}(b)}}{\abs{\mc{A}}}
		-
		\frac{1}{\abs{\mc{B}}}
		\right)
		g(b)
	}
	=
	\abs{
	\sum_{b \in \mc{S}}
	\left(
		\frac{\abs{f^{-1}(b)}}{\abs{\mc{A}}}
		-
		\frac{1}{\abs{\mc{B}}}
	\right)
	g(b)
	+
	\sum_{b \in \mc{S}'}
	\left(
		\frac{\abs{f^{-1}(b)}}{\abs{\mc{A}}}
		-
		\frac{1}{\abs{\mc{B}}}
	\right)
	g(b)
	}.
\end{equation}
Note that the $\mc{S}$ term in the right-hand side of the above equation
is positive and the $\mc{S}'$ term is negative. Recalling that $g(b)
\leq 1$, we then have
\begin{equation}
	\abs{
		\sum_{b \in \mc{B}}
		\left(
		\frac{\abs{f^{-1}(b)}}{\abs{\mc{A}}}
		-
		\frac{1}{\abs{\mc{B}}}
		\right)
		g(b)
	}
	\leq
	\max\left\{
	\sum_{b \in \mc{S}}
	\left(
		\frac{\abs{f^{-1}(b)}}{\abs{\mc{A}}}
		-
		\frac{1}{\abs{\mc{B}}}
	\right)
	,
	\sum_{b \in \mc{S}'}
	\left(
		\frac{1}{\abs{\mc{B}}}
		-
		\frac{\abs{f^{-1}(b)}}{\abs{\mc{A}}}
	\right)
	\right\}.
\end{equation}
Lastly, we note that
\begin{equation}
	\sum_{b \in \mc{S}}
	\left(
		\frac{\abs{f^{-1}(b)}}{\abs{\mc{A}}}
		-
		\frac{1}{\abs{\mc{B}}}
	\right)
	=
	\underbrace{
	\frac{1}{2}
	\sum_{b \in \mc{B}}
	\abs{
		\frac{\abs{f^{-1}(b)}}{\abs{\mc{B}}} - \frac{1}{\abs{\mc{B}}}
	}
	}_{= \epsilon}
	=
	\sum_{b \in \mc{S}'}
	\left(
		\frac{1}{\abs{\mc{B}}}
		-
		\frac{\abs{f^{-1}(b)}}{\abs{\mc{A}}}
	\right)
\end{equation}
which follows from the fact that $
	\sum_{b \in\mc{S}}\abs{f^{-1}(b)}
	=
	\abs{\mc{A}} - \sum_{b \in\mc{S}'} \abs{f^{-1}(b)}
$ and direct calculations.
\end{proof}

\begin{lemma}
\label{th:eps-existence}
For any finite sets $\mc{A}$ and $\mc{B}$, there exists an
$\abs{\mc{B}}/(4\abs{\mc{A}})$-uniform map $f : \mc{A} \to \mc{B}$.
\end{lemma}
\begin{proof}
Assume that $\mc{A} = \{0, \ldots, \abs{\mc{A}}-1\}$,
$\mc{B} = \{0, \ldots, \abs{\mc{B}}-1\}$ and take $f$ to be
$x \mapsto x \pmod{\abs{\mc{B}}}$. Let
$r = \abs{\mc{A}}/\abs{\mc{B}}$ and
$\ell = \abs{\mc{A}} - \floor{r} \abs{\mc{B}}$. We then have that
\begin{equation}
\begin{split}
	\frac{1}{2}\sum_{b \in
	\mc{B}}\abs{\frac{\abs{f^{-1}(b)}}{\abs{\mc{A}}} -
	\frac{1}{\abs{\mc{B}}}}
	&=
	\frac{1}{2}
	\left(
		\ell
		\cdot
		\abs{\frac{\floor{r}+1}{\abs{\mc{A}}} - \frac{1}{\abs{\mc{B}}}}
		+
		(\abs{\mc{B}}-\ell) \cdot
		\abs{\frac{\floor{r}}{\abs{\mc{A}}} - \frac{1}{\abs{\mc{B}}}}
	\right)
	\\&=
	1
	-
	r
	+
	2\floor{r} - \frac{\floor{r}}{r} - \frac{\floor{r}^2}{r}
\end{split}
\end{equation}
where we can remove the absolute values by noting that
$\abs{\mc{B}}(\floor{r}+1) - \abs{\mc{A}} \geq 0$ and
$\abs{\mc{B}}\floor{r} - \abs{\mc{A}} \leq 0$. Letting $r = w + p$ for
$w \in \N$ and $0 \leq p < 1$ and noting that $\floor{r} = w$, we have
that
\begin{equation}
	1
	-
	r
	+
	2\floor{r} - \frac{\floor{r}}{r} - \frac{\floor{r}^2}{r}
	=
	\frac{p - p^2}{r}
	\leq
	\frac{1}{4r}
\end{equation}
where the equality is obtained by direct calculation and the inequality
by noting that $p - p^2 \leq \frac{1}{4}$.
\end{proof}

\begin{lemma}
\label{th:QAS-key-change}
Let $
	\textsf{S}
	=
	\{\left(\textsf{Auth}_k, \textsf{Ver}_k\right)\}_{k \in \mc{K}}
$ be an $\epsilon$-total QAS. Let $\mc{K}'$ be a finite set and
$f : \mc{K}' \to \mc{K}$ be an $\epsilon'$-uniform map. Finally, for
every $k' \in \mc{K}'$, we define
\begin{equation}
	\textsf{fAuth}_{k'} = \textsf{Auth}_{f(k')}
	\qq{and}
	\textsf{fVer}_{k'} = \textsf{Ver}_{f(k')}.
\end{equation}
Then, $
	\textsf{fS}
	=
	\{\left(\textsf{fAuth}_k, \textsf{fVer}_k\right)\}_{k \in \mc{K}'}
$ is an $(\epsilon + \epsilon')$-total QAS.
\end{lemma}
\begin{proof}
Let $\Phi$ be an attack against the $\textsf{fS}$ scheme. Then, it is
also a valid attack against the $\textsf{S}$ scheme. As $\textsf{S}$ is
an $\epsilon$-total QAS, there exists a completely positive
trace non-increasing map $\Psi$ such that
\begin{equation}
\label{eq:fauth-1}
	\E_{k \in \mc{K}}
	\ketbra{k}
	\tensor
	\textsf{Ver}'_k \circ \Phi \circ \textsf{Auth}_k
	(\rho)
	\approx_\epsilon
	\E_{k \in \mc{K}}
	\ketbra{k}
	\tensor
	\textsf{Ver}'_k \circ \Psi \circ \textsf{Auth}_k
	(\rho)
\end{equation}
for all $\rho \in \mc{D}(\tsf{MZS})$. It then suffices to show that
\begin{equation}
\label{eq:fauth-2}
	\E_{k' \in \mc{K}'}
	\ketbra{k'}
	\tensor
	\textsf{fVer}'_{k'} \circ \Phi \circ \textsf{fAuth}_{k'}
	(\rho)
	\approx_{\epsilon + \epsilon'}
	\E_{k' \in \mc{K}'}
	\ketbra{k'}
	\tensor
	\textsf{fVer}'_{k'} \circ \Psi \circ \textsf{fAuth}_{k'}
	(\rho)
\end{equation}
for all $\rho \in \mc{D}(\tsf{MZS})$ to prove the claim.

Fix a state $\rho$ and, to lighten notation, define the operators
\begin{equation}
\begin{split}
	\alpha_k
	&= 
	\textsf{Ver}'_k \circ \Phi \circ \textsf{Auth}_k(\rho)
	\text{,}\hspace{2.1cm}
	\beta_k
	= 
	\textsf{Ver}'_k \circ \Psi \circ \textsf{Auth}_k(\rho),
	\\
	\phi_{k'}
	&= 
	\textsf{fVer}'_{k'} \circ \Phi \circ \textsf{fAuth}_{k'}(\rho),
	\qq{and}
	\varphi_{k'}
	= 
	\textsf{fVer}'_{k'} \circ \Psi \circ \textsf{fAuth}_{k'}(\rho).
\end{split}
\end{equation}
Next, note that $\phi_{k'} = \alpha_{f(k')}$,
$\varphi_{k'} = \beta_{f(k')}$, and that by
\cref{th:trace-distance-orthogonal}, the
inequalities in \cref{eq:fauth-1,eq:fauth-2} are equivalent to
\begin{equation}
	\E_{k} \Delta(\alpha_k, \beta_k) \leq \epsilon
\qq{and}
	\E_{k'} \Delta(\phi_{k'}, \varphi_{k'}) \leq \epsilon + \epsilon'
\end{equation}
respectively. Finally, by \cref{th:eps-expectation}, we have that
\begin{equation}
\begin{split}
	\abs{
		\E_{k} \Delta(\alpha_k, \beta_k)
		-
		\E_{k'} \Delta(\phi_{k'}, \varphi_{k'})
	}
	=
	\abs{
		\E_{k} \Delta(\alpha_k, \beta_k)
		-
		\E_{k'} \Delta(\alpha_{f(k')}, \beta_{f(k')})
	}
	\leq
	\epsilon'
\end{split}
\end{equation}
which implies that $
	\E_{k'} \Delta(\phi_{k'}, \varphi_{k'})
	\leq
	\E_{k} \Delta(\alpha_k, \beta_k)
	+
	\epsilon'
	\leq
	\epsilon + \epsilon'
$
which completes the proof.
\end{proof}

Next, we recall a theorem which states that any unitary $2$-design can
be used to construct a total QAS. For our needs, it suffices to know
that a unitary $2$-design on $n$ qubits is a set of unitary operators on
$n$-qubits satisfying certain conditions (see \cite{DCEL09}).

\begin{theorem}[{\cite{AM17}}]
\label{th:AM17}
A unitary $2$-design on $n+t$ qubits can be used to construct a
$\left(2^{\frac{6-t}{3}}\right)$-total QAS for $n$ qubits where the key
set is the unitary $2$-design.
\end{theorem}

\begin{theorem}[{\cite{CLLW16}}]
\label{th:2-design}
There exists a unitary $2$-design on $n$ qubits with $2^{5n} - 2^{3n}$
elements.
\end{theorem}

Finally, we give the proof of \cref{th:QAS-existence}. Recall that
\cref{th:QAS-existence} state that for any strictly positive $n$ and
$k$, there exists a $\left(5\cdot2^{\frac{5n-k}{16}}\right)$-total QAS
on $n$ qubits with key set $\{0,1\}^k$.

\begin{proof}[Proof of \cref{th:QAS-existence}]
We consider two cases: when $k \geq 5a + 38$ and when $k < 5a + 38$.

If $k < 5a + 38$, then $5\cdot2^{\frac{5a-k}{16}} \geq 1$ and so it
suffices to find a $1$-QAS on $n$ qubits. Taking the authentication
and verification maps to be the identity (with an extra output of an
accept flag in the verification map) for every key is sufficient.

We now consider the non-trivial case of $k \geq 5a + 38$. From
\cref{th:2-design,th:AM17}, we have the existence of a
$\left(2^{\frac{6-t}{3}}\right)$-QAS on $n$ qubits with a key set of
size $2^{5(n+t)} - 2^{3(n+t)}$. From \cref{th:eps-existence},
there exists an $\epsilon'$-uniform map from $\{0,1\}^k$ to this key set
for
\begin{equation}
	\epsilon'
	=
	\frac{1}{4} \cdot \frac{2^{5(a+t)} - 2^{3(a+t)}}{2^k}
\end{equation}
Thus, by \cref{th:QAS-key-change}, there exist an $\epsilon$-QAS on $a$
qubits with key set $\{0,1\}^k$ for
\begin{equation}
	\epsilon
	\leq
	2^{\frac{6-t}{3}}
	+
	\frac{1}{4} \cdot \frac{2^{5(a+t)} - 2^{3(a+t)}}{2^k}
	=
	2^{2-\frac{t}{3}}
	+
	2^{5a + 5t - k -2}
	-
	2^{3a + 3t - k - 2}
	<
	2^{2-\frac{t}{3}}
	+
	2^{5a + 5t - k -2}
\end{equation}
where we obtain the last inequality by simply dropping the third term
which will allow us to simplify our upcoming analysis. It then suffices
to find the optimal value of $t$ for given and fixed $a$ and $k$.

Direct methods (i.e. first and second derivatives) yield that the
optimal value of $t$ to minimize the above upper bound is
\begin{equation}
	t_\text{opt} = \frac{12 - 3 \log_2(15) + 3k - 15a}{16}.
\end{equation}
However, from \cref{th:AM17}, $t$ must be a
non-negative integer. Thus, we take $t = \floor{t_\text{opt}}$.
The condition that $k \geq 5a + 38$ also yields that $t$ is strictly
positive.

Finally, substituting this choice of $t$ in our upper bound of
$\epsilon$, we have that
\begin{equation}
	\epsilon
	<
	2^{2 - \frac{\floor{t_\text{opt}}}{3}}
	+
	2^{5a + 5\floor{t_\text{opt}} - k - 2}
	\leq
	2^{2 - \frac{t_\text{opt}-1}{3}}
	+
	2^{5a + 5t_\text{opt} - k - 2}
	=
	\frac{2^{\frac{23}{4}}}{15^\frac{15}{16}}\cdot2^{\frac{5a-k}{16}}
	<
	5 \cdot 2^{\frac{5a-k}{16}}
\end{equation}
which is the desired result.
\end{proof}

%========================================%
\section{Proofs for {\cref{sec:SSLPF-to-SSLCC}}} %
\label{sc:gen-proofs}                    %
%========================================%
In this appendix, we give the proof of \cref{thm:SSLPF-to-SSLCC}. We do this
in two parts: we first prove correctness, followed by security.

\begin{lemma}
If for all $f\in F$, the scheme $\SSLPF$ is $\eta$-correct with respect
to the family of distributions $\{T_{f,y}^{PF}\}_{y\in\{0,1\}^m}$, then
the scheme $\SSLCC$ described in \cref{sec:SSLPF-to-SSLCC} is $\eta$-correct with respect to the family of distributions $\{T_{f,y}^{CC}\}_{(f,y)\in \calF}$. 
\end{lemma}

\begin{proof}
The proof proceeds by considering the definition of correctness for $\SSLCC$ and $\SSLPF$ from~\Cref{def:ssl-correct}. Let $\sk$ denote the secret key generated by $\SSLCCGen=\SSLPFGen$. For $(f,y) \in \calF$, observe that $\SSLCCLease(\sk,(f,y))$ outputs the program $(f, \rho)$ where $\rho = \SSLPFLease(\sk, P_y)$. Thus the probability that $\SSLCCEval$, given the program $(f,\rho)$ and input $x\leftarrow T_{f,y}^{CC}$, outputs the correct value $\CC_y^f(x)=P_y(f(x))$ is:
	\begin{equation}
\begin{split}
		& \E_{x \leftarrow T_{f,y}^{CC}} \Tr\left( \ketbra{\CC_y^f(x)}\SSLCCEval\left( (f, \rho), x \right)\right) \\
={}& \E_{x \leftarrow T_{f,y}^{CC}} \Tr\left( \ketbra{P_y(f(x))}\SSLPFEval\left(\rho, f(x) \right)\right) \\
={}& \E_{x \leftarrow T_{f,y}^{CC}} \Tr\left( \ketbra{P_y(f(x))}\SSLPFEval\left(\SSLPFLease(\sk,P_y), f(x) \right)\right) 	\\
={}& \E_{z \leftarrow T_{f,y}^{PF}} \Tr\left( \ketbra{P_y(z)}\SSLPFEval\left(\SSLPFLease(\sk,P_y), z\right)\right)
 \geq 1 - \eta,
	\end{split}\label{eq:corrSSLCC1}
\end{equation}
	where the third equality used the definition of $T_{f,y}^{PF}$ with $z = f(x)$ and the last inequality uses that $\SSLPF$ is $\eta$-correct with respect to $\{T_{f,y}^{PF}\}_{y\in\{0,1\}^m}$.

	To satisfy the correctness requirement for $\SSLCCVerify$, note that the probability that $\SSLCCVerify$ accepts $\rho=\SSLPFLease(\sk, P_y)$ is:
	 \begin{align}
	 \label{eq:corrSSLCC2}
		\Tr(\ketbra{1}\SSLCCVerify(\sk,  (f, y), \rho)) = \Tr(\ketbra{1}\SSLPFVerify(\sk, P_y, \rho)) \geq 1 - \eta,
	\end{align}
	where the last inequality uses the $\eta$-correctness of $\SSLPF$. Putting together~\Cref{eq:corrSSLCC1,eq:corrSSLCC2}, we can conclude that $\SSLCC$ is $\eta$-correct with respect to the ensemble $\{T_{f,y}^{CC}\}_{(f,y)\in\calF}$. 
\end{proof}

We now move to the security guarantees for the scheme $\SSLCC$. Observe that, the program distribution $D$ for compute-and-compare programs is a distribution on $(f,y)\in \calF$. For every $f$, we obtain a marginal distribution on $y$, which we denote by $D_f$. These could be different in each case. Similarly, the set of challenge distributions for compute-and-compare functions, $\{D_{f,y}^{CC}\}_{(f,y)\in\calF}$, has a distribution over $\{0,1\}^n$ for each compute-and-compare function $(f,y)\in\calF = F\times\{0,1\}^m$. A set of challenge distributions for point functions would have a distribution over $\{0,1\}^m$ for each point function $P_y$ such that $y\in\{0,1\}^m$. For every fixed $f$, we can obtain such a distribution $D_{f,y}^{PF}$ by sampling $x\leftarrow D_{f,y}^{CC}$ and outputting $f(x)$ (as defined in \cref{thm:SSLPF-to-SSLCC}). We require the scheme $\SSLPF$ to be secure against program distribution $D_f$ and challenge distributions $\{D_{f,y}^{PF}\}_{y\in\{0,1\}^m}$ for \emph{every} $f\in F$. 
We then get the following theorem showing the security of the scheme
$\SSLCC$ constructed in \cref{sec:SSLPF-to-SSLCC}. 

\begin{lemma}[Restatement Theorem 6 of \cite{CMP20}]
If for all $f\in F$ the scheme $\SSLPF$ is $\epsilon_f$-secure (with
$\epsilon_f$ defined in \cref{eq:eps_f}), with respect to program
distribution $D_f$ and challenge distributions
$\{D_{f,y}^{PF}\}_{y\in\{0,1\}^m}$, then the scheme $\SSLCC$ described
in \cref{sec:SSLPF-to-SSLCC} is $\epsilon$-secure with respect to $D$ and $\{D_{f,y}^{CC}\}_{(f,y)\in\calF}$. 
\end{lemma}

Note that in \cite{CMP20} the above theorem was presented for a specific
family of distributions. Our restatement is applicable to all
distributions. We provide the proof below for completeness.

\begin{proof}
The proof proceeds via contradiction by showing that an adversary $\adv_{\CC}$ that can win the game $\ExperimentSSL_{\adv_{\CC},\SSLCC}$ for compute-and-compare functions can be used to construct an adversary $\adv_{\PF}$ that can win the game $\ExperimentSSL_{\adv_{\PF},\SSLPF}$ for point functions.

Assume that $\adv_{\CC}$ can win the game $\ExperimentSSL_{\adv_{\CC},\SSLCC}$ instantiated with the program distribution $D$ and challenge distributions $\{D_{f, y}^{CC}\}_{(f,y)\in\calF}$ for the scheme $\SSLCC$ with probability $>p^{\text{triv}}_{D,\{D^{CC}_{f,y}\}_{(f,y)}}+\epsilon$.  
Then there exists some $f^*\in F$ such that the probability that $(f^*,y)$ for some $y$ is sampled from $D$ is non-zero, and 
$\adv_{\CC}$ wins the game $\ExperimentSSL$ with probability $>p^{\text{triv}}_{D,\{D^{CC}_{f,y}\}_{(f,y)}}+\epsilon$ conditioned on this event.

We construct an adversary $\adv_{\PF}$ for $\SSLPF$ that wins the
$\ExperimentSSL_{\adv_{\PF},\SSLPF}$ instantiated with the program
distribution $D_{f^*}$ and challenge distributions
$\{D_{f^*,y}^{PF}\}_{y\in\{0,1\}^m}$ with probability
strictly greater than $p^{\text{triv}}_{D_{f^*},\{D^{PF}_{f^*,y}\}_{y}}+ \epsilon_{f^*}$, which is a contradiction. 
The behaviour of $\adv_{\CC}$ can be described in two parts  (see also \cref{fig:SSL-experiment}). 
First, the adversary applies an arbitrary CPTP map $\Phi_{\adv_{\CC}}:{\cal L}(\reg{Y})\rightarrow{\cal L}(\reg{YA})$ to his input $\rho$, sending the $\reg{Y}$ part to the Lessor 
(Step 2 of $\ExperimentSSL$); and keeping the $\reg{A}$ part for himself. Later, when $\adv_{\CC}$ receives the challenge $x$, he uses it to select a two outcome measurement $\{\Pi_x,I-\Pi_x\}$ with which to measure his register $\reg{A}$ to obtain a bit $b$ (Step 5 of $\ExperimentSSL$). 
Construct an adversary $\adv_{\PF}$ so that the game $\ExperimentSSL_{\adv_{\PF},\SSLPF}$ proceeds as follows:
\begin{itemize}
\item The Lessor samples $y\leftarrow D_{f^*}$ and runs $\SSLPFGen$ to obtain a secret key $\sk$. She then sends $\rho=\SSLPFLease(\sk,P_y)$ to $\adv_{\PF}$.
\item $\adv_{\PF}$ makes use of $\Phi_{\adv_{\CC}}$, which expects, as input, an encoded compute-and-compare program output by $\SSLCCLease$. Such a program has the form $(f,\xi)$, which we use as a shorthand for $\ketbra{f}\otimes\xi$, for $f\in F$ and $\xi\in{\cal D}(\reg{Y})$, an encoded point function output by $\SSLPFLease$. In other words, the output space of $\SSLCCLease$ is $\reg{Y}'=\reg{FY}$ where $\reg{F}=\mathrm{span}\{\ket{f}:f\in F\}$. 
The pair $(f^*,\rho)$ fits this description. $\adv_{\PF}$
computes $\sigma=\Phi_{\adv_{\CC}}(f^*,\rho)\in{\cal D}(\reg{Y}'\reg{A})={\cal D}(\reg{FYA})$ and sends the $\reg{Y}$ part of $\sigma$ back to the Lessor and keeps the $\reg{FA}$ part. 
\item The Lessor runs $\SSLPFVerify(\sk,P_y,\cdot)=\SSLCCVerify(\sk,(f^*,y),\cdot)$ on $\reg{Y}$ and aborts if the resulting bit $v$ is not 1.
\item The Lessor samples $z\leftarrow D_{f^*,y}^{PF}$, which is an $m$-bit string, and sends $z$ to $\adv_{\PF}$.
\item $\adv_{\PF}$ samples $x$ according to the restriction of $D_{f^*,y}^{CC}$ to the set of pre-images $f^{*-1}(z)$.\footnote{The adversary being computationally unbounded has sufficient resources to construct ${f^*}^{-1}$ given $f^*$.} He then
measures his register $\reg{A}$ using the two-outcome measurement $\{\Pi_x,I-\Pi_x\}$ and returns the resulting bit $b$ to the Lessor.
\item The Lessor outputs 1 if and only if $b=P_y(x)$ and $v=1$. 
\end{itemize}
Let $\Pi_x^1=\Pi_x$ and $\Pi_x^0=I-\Pi_x$. Let $\Psi_{\sf Ver}^{P_y}$ be the map induced by $\SSLPFVerify(\sk,P_y,\cdot)$, and $\Psi_{\sf Ver}^{\CC^{f^*}_y}$ the map induced by $\SSLCCVerify(\sk,(f^*,y),\cdot)$, and note that these are the same map. 
The winning probability is given by:
\begin{equation}
\begin{split}
\Pr[\ExperimentSSL_{\adv_{\PF},\PF}]=&\sum_{y,z\in\{0,1\}^m}D_{f^*}(y)D_{f^*,y}^{PF}(z)\Tr\left( (\ketbra{1}\otimes \Pi_x^{P_y(z)})(\Psi_{\sf Ver}^{P_y}\otimes\Id_{\sf FA})\circ\Phi_{\adv_{\CC}}(f^*,\rho) \right).
\end{split}
\end{equation}
Notice that the string $x$ is distributed according to $D_{f^*,y}^{CC}$, so we can rewrite this probability as:
\begin{equation}
\begin{split}
 \sum_{y\in\{0,1\}^m,x\in\{0,1\}^n}D_{f^*}(y)D_{f^*,y}^{CC}(x)\Tr\left( (\ketbra{1}\otimes \Pi_x^{P_y(f(x))})(\Psi_{\sf Ver}^{\CC^{f^*}_y}\otimes\Id_{\sf FA})\circ\Phi_{\adv_{\CC}}(f^*,\rho) \right)\\
={}\sum_{y\in\{0,1\}^m,x\in\{0,1\}^n}D_{f^*}(y)D_{f^*,y}^{CC}(x)\Tr\left( (\ketbra{1}\otimes \Pi_x^{\CC_y^f(x)})(\Psi_{\sf Ver}^{\CC^{f^*}_y}\otimes\Id_{\sf FA})\circ\Phi_{\adv_{\CC}}(f^*,\rho) \right).
\end{split}
\end{equation}
Finally, note that $(f^*,\rho)=\SSLCCLease(\sk,(f^*,y))$, so this is
exactly the probability of $\adv_{\CC}$ winning the game
$\ExperimentSSL_{\adv_{\CC},\CC}$
conditioned on $f^*$ being the sampled function, which is
$>p^{\text{triv}}_{D,\{D^{CC}_{f,y}\}_{(f,y)}}+\epsilon$ by assumption.
Thus,
\begin{equation}
\begin{split}
	\Pr[\ExperimentSSL_{\adv_{\PF},\PF}]
	&>
	p^{\text{triv}}_{D,\{D^{CC}_{f,y}\}_{(f,y)}}+\epsilon
	\\&=
    p^{\text{triv}}_{D_{f^*},\{D^{PF}_{f^*,y}\}_{y}}
	+
	\left(
		p^{\text{triv}}_{D,\{D^{CC}_{f,y}\}_{(f,y)}}
		-
		p^{\text{triv}}_{D_{f^*},\{D^{PF}_{f^*,y}\}_{y}}
	\right)
	+
	\epsilon
	\\&=
	p^{\text{triv}}_{D_{f^*},\{D^{PF}_{f^*,y}\}_{y}}
	+
	\epsilon_{f^*}	
\end{split}
\end{equation}
which is a contradiction.
\end{proof}

\bibliographystyle{bibtex/bst/alphaarxiv}
\bibliography{bibtex/bib/full,bibtex/bib/quantum,bibtex/bib/quantum-more}

%%%%%%%%%%%%%%%%%%
%%%%%%%%%%%%%%%%%%
\end{document}